\documentclass[10pt]{elsarticle}
\usepackage{graphicx}
\usepackage{balance}
\usepackage{algorithm}
\usepackage{algpseudocode}
\usepackage{url}
\usepackage{hyphenat}
\usepackage{geometry} 
\usepackage{amssymb}
\usepackage{amsthm}
\usepackage{amsmath}
\usepackage{epstopdf}
\usepackage{epsfig}
\usepackage{color}
\usepackage[small,bf]{caption}

\newtheorem{thm}{Theorem}

\newtheorem{defn}{Definition}
\newtheorem{lem}{Lemma}
\newtheorem{cor}{Corollary}

\newtheorem{claim}{Claim}

\begin{document}

%\author{
%	Michael Simpson\\
%	\texttt{simpsonm@uvic.ca}
%	\and
%	Venkatesh Srinivasan\\
%	\texttt{venkat@cs.uvic.ca}
%	\and
%	Alex Thomo\\
%	\texttt{thomo@cs.uvic.ca}
%}

%\maketitle

\begin{frontmatter}

\title{Clearing Contamination in Large Networks}

\author[us]{Michael Simpson}
\ead{simpsonm@uvic.ca}

\author[us]{Venkatesh Srinivasan}
\ead{venkat@cs.uvic.ca}

\author[us]{Alex Thomo}
\ead{thomo@cs.uvic.ca}

\address[us]{University of Victoria}

\begin{abstract}
In this work, we study the problem of clearing contamination spreading through a large network where we model the problem as a graph searching game. The problem can be summarized as constructing a search strategy that will leave the graph clear of any contamination at the end of the searching process in as few steps as possible. We show that this problem is NP-hard even on directed acyclic graphs and provide an efficient approximation algorithm. We experimentally observe the performance of our approximation algorithm in relation to the lower bound on several large online networks including Slashdot, Epinions and Twitter. The experiments reveal that in most cases our algorithm performs near optimally.
\end{abstract}


\end{frontmatter}

\section{Introduction}

%There exists a large body of recent work which studies how information or influence propagates through a social network which is applicable to other types of networks such as those mentioned previously. Specifically, \cite{budak2011limiting, Meier2008, Bharathi2007, Carnes2007} focused on the task of limiting the spread of misinformation in a social network while we study the stronger model of eliminating disinformation, or any kind of contamination, from a general network.

Contamination in a network may refer to information propagating through an online social network, viruses spreading through a water network, sickness spreading though a population, or a number of other situations. In particular, we are interested in studying social networks as they allow for the widespread distribution of knowledge and information in modern society. They are rapidly becoming a place where people go to hear the news and discuss personal and social topics. In turn, the information posted can spread quickly through the network eventually reaching a large audience, especially so for influential users. However, information spread in a social network can have either positive or negative effects. For example, posting about natural disasters or warfare can either help or hinder other users depending on whether the information is accurate or not. In other cases, the information can be strictly detrimental, such as defamatory statements about private corporations or people and rumours negatively affecting the financial markets. Thus, since many people today learn of news or events online it is important to have tools to eliminate, not just minimize, the effects of disinformation. Previous work has focused on the task of limiting the spread of misinformation \cite{budak2011limiting, Meier2008, Bharathi2007, Carnes2007} while we study the stronger model of eliminating disinformation, or any kind of contamination, from a general network. 

For a contaminated network, we model the problem in the context of graph searching; a classical game on graphs \cite{Parsons1976, Borie2011, Dendris1994}. In the graph searching game we may think of a network whose edges are contaminated with a gas and the objective is to clean the network with some number of searchers. However, the gas immediately recontaminates cleared edges if its spreading is not blocked by guards at the vertices. The model does not assume knowledge of the location of the gas, yet guarantees its elimination at the end of the search strategy, and assumes an edge is deterministically contaminated, as opposed to probabilistically, which represents the case of a powerful adversary.

In the pioneering work of Brandenburg and Herrmann \cite{Brandenburg2006UQ} the dual to the well studied \emph{search number} (the minimum number of searchers required to clear a graph), search time, was introduced as a new cost measure in graph searching. Naturally, we believe it is more important to clear the network as quickly as possible when dealing with a contaminant. Furthermore, until now the theory community has mainly focused on the search number of an undirected graph, but one needs to study the more general case of directed graphs as many real world networks lend themselves to be modelled as directed.

We study the problem of minimizing the time required to eliminate the contamination in the network given a budget of searchers. We prove that the search time problem is NP-complete even for directed acyclic graphs (DAGs) and introduce an approximation algorithm for clearing DAGs. Furthermore, we propose a method for clearing a network by first reducing it to a DAG which can be cleared by our approximation algorithm. Additionally, we investigate the merits of a split and conquer style strategy and show that our strategy, which instead has searchers staying together as a group, outperforms the (intuitively appealing) split and conquer strategy on a broad class of DAGs. Along the way we prove lower bounds on the time required to search a directed graph and introduce a novel DAG decomposition theorem.

We note that the study of search time is intrinsically more difficult than computing the search number as we can no longer be ``strategy oblivious''. By that we mean, when studying the search number, one is only interested in knowing whether some search strategy exists to clear a graph with some number of searchers. In this setting, how that strategy works is irrelevant to the end goal. In contrast, trying to compute the search time of a graph is closely tied to how the strategy actually plays out.

Our main contributions can be summarized as follows. 

\begin{enumerate}
\item We are the first to investigate the search time of directed graphs.
\item We prove the search time problem is NP-complete on DAGs.
\item We devise an approximation algorithm for clearing DAGs that also outperforms split and conquer strategies on a broad class of DAGs.
\item We introduce a novel DAG decomposition theorem which we believe is of independent interest.
\item We provide an experimental study of clearing large real and synthetic networks.
\end{enumerate}

We start with an overview of information propagation in social networks and the graph searching problem in Section 2. In Section 3 we introduce the necessary concepts and definitions from graph searching. Section 4 presents the lower bound for search time on directed graphs. In Section 5 we prove the NP-hardness of the search time problem on DAGs. We introduce our strategy for clearing general networks and the Plank algorithm in Section 6. Section 7 contains our approximation bounds, comparison to the split and conquer strategy, along with our DAG decomposition theorem. Finally, in Section 8 we provide our experimental results.

\section{Related Work}

The task of maximizing the spread of information in a social network is a well studied problem with many works investigating different aspects of the problem \cite{Chen2013, Goyal2013, Chen2010, Liu2012}.  More recently, the problem of limiting the spread of rumours or misinformation in a social network has been studied by \cite{budak2011limiting, Bharathi2007, Carnes2007}. In \cite{Bharathi2007, Carnes2007} the problem is posed in terms of competing campaigns while \cite{budak2011limiting} has the misinformation diffusing through a network. All three works are modelled by the Independent Cascade Model: a randomized diffusion process on graphs. However, the location of the misinformation is known and nodes can be inoculated such that once a node takes on the ``good'' information it will not subsequently adopt the misinformation. While the goal of these works was to limit the spread of misinformation, we believe it is important to investigate how to remove the misinformation from a network in its entirety. Furthermore, the unknown location of the misinformation and the deterministic spreading of contamination in our model captures the case of a stronger adversary.

Several variants of the (undirected) graph searching problem with respect to search number have been studied with varying constraints and adversary behaviour, see e.g. \cite{Dendris1994, Blin2008DCN, Borie2011, Kirousis1986SP, Ellis1994}. In addition, it has been shown that the graph searching problem is closely related to several other notable graph parameters such as path-width, cut-width and vertex separation, see e.g. \cite{Bienstock1991, Kirousis1986SP, Ellis1994}. It was shown by Megiddo et al.\ \cite{Megiddo1988} that computing the search number is NP-complete on general undirected graphs, but can be computed in linear time on undirected trees.

%In the directed setting, the notion of \emph{DAG-width} was introduced by Berwanger et al.\ \cite{Berwanger2012} as a natural adaptation of treewidth (which characterizes search number) to directed graphs. They study the search number problem for directed graphs and prove similar results to those for treewidth.

The notion of search time for undirected graphs was introduced by Brandenburg and Herrmann \cite{Brandenburg2006UQ}. They note that the classical goal of the graph searching game where the minimal search number is computed aims to minimize the number of resources used and as such corresponds to space complexity. They study the length of a search strategy which corresponds to the time complexity of searching a graph. They ask, how fast can a team of $k$ searchers clear a graph (if at all), and conversely how many searchers are needed to search a graph in time $t$. %The objective of that paper is to generalize known results regarding search number to search time. They are successful in showing monotone graph searching as well as establishing an equivalence between search time, pathlength, interval length, and vertex length.

\section{Preliminaries}

We consider the graph searching game on simple, weakly connected, directed graphs $G = (V, E)$ with $n$ nodes, a set of vertices $V$ and a set of edges $E$. We assume there are no self-loops and no multiple edges. A directed graph is considered weakly connected if removing the directions on all edges yields an undirected graph which is connected. For a directed edge $(u,v)$ we refer to $u$ as the start node and $v$ as the end node. Also, we will use the term ``digraph'' when referring to directed graphs.

The rules for the graph searching game are as follows: Initially, all edges are contaminated and in the end all edges must be cleared. In a move at each time $t = 1, 2, \dots$ searchers (or guards) are \emph{first} removed from vertices and then placed on other (and possibly the same) vertices. In a single move some number of searchers can be placed or removed subject to the searcher budget. An edge is \emph{cleared} at time $i$ if both incident nodes have searchers placed on them at the end of time $i$. A cleared edge $e$ is instantaneously \emph{recontaminated} if there is a directed path from a contaminated edge to $e$ without a searcher on any vertex of that path. A \emph{search strategy} is a sequence of moves that results in all edges being cleared at the end. Then the search game is won.

In the following example we show one possible search strategy with four available searchers for the directed graph shown in Figure~\ref{fig:example_strat}. In the first step, searchers are placed on nodes 1, 2, 3, and 4 clearing the three blue (double-wide) edges. In the second step, searchers are removed from nodes 1, 2, and 3 to be placed on nodes 5, 8, and 6. We clear another three edges, and mark cleared edges in green (dotted). Finally, in a third step, we remove searchers from nodes 4 and 5, and place them on nodes 7 and 9. We clear the final three edges in the third step leaving the graph with all its edges cleared.

\begin{figure}
	\centering
	\includegraphics[scale=0.5]{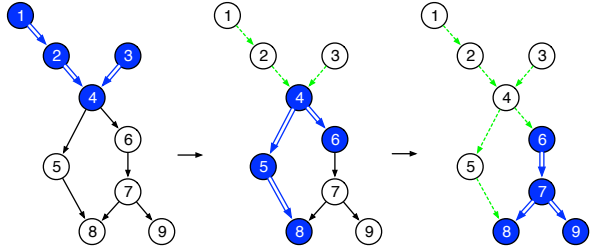}
	\caption{An example search strategy}
	\label{fig:example_strat}
\end{figure}

Our formal definition is similar to that of Brandenburg and Herrmann \cite{Brandenburg2006UQ}.

\begin{defn}
A search strategy $\sigma$ on a (connected) digraph $G=(V,E)$ is a sequence of pairs $\sigma = ((E_0, V_0),\\ (E_1, V_1), \dots, (E_t, V_t))$ such that:

\begin{enumerate}
\item For $i = 0, \dots, t$, $E_i \subseteq E$ is the set of cleared edges and $V_i \subseteq V$ is the set of vertices which have searchers placed on them at time $i$. The edges from $E \setminus E_i$ are contaminated.
\item (initial state) $E_0 = \emptyset$ and $V_0 = \emptyset$. All edges are contaminated.
\item (final state) $E_t = E$ and $V_t = \emptyset$. All edges are cleared.
\item (remove and place searchers and clear edges) For $i = 0, \dots, t-1$ there are sets of vertices $R_i = V_i \setminus V_{i+1}$ and $P_i = V_{i+1} \setminus V_i$ where searchers are removed from the vertices from $R_i$ and then placed at $P_i$. The set of cleared edges is $E_{i+1} = \{ (u, v) \in E \: | \: u, v \in V_{i+1}$; or $(u, v) \in E_i \: |$ there is no unguarded directed path from the end node of a contaminated edge to u$\}$.
\end{enumerate}

Let $width(\sigma) = max\{ | V_i | \: | \: i = 0, \dots, t \}$ and $length(\sigma) = t - 1$ be the number of searchers and the number of moves of $\sigma$ respectively. Note that we discard the last move, which only removes searchers.
\end{defn}

While we need the $E_i$ sets above to define how a strategy works, we only need the $V_i$ sets to fully determine a strategy. Therefore, we will often refer to a strategy by only listing its $V_i$ sets.

\begin{defn}
For a connected digraph $G$ with at least two vertices and integers $s$ and $t$ let search-width$_G(t)$ be the least $width(\sigma)$ for all search strategies $\sigma$ with $length(\sigma) \leq t$ and let search-time$_G(s)$ be the least $length(\sigma)$ for all search strategies $\sigma$ with $width(\sigma) \leq s$.
\end{defn}

In other words, \emph{search-width$_G(t)$} is the least number of searchers that can search $G$ in time at most $t$, and \emph{search-time$_G(s)$} is the shortest time such that at most $s$ searchers can search $G$. Thus, \emph{search-width}$_G(t) = s$ implies \emph{search-time}$_G(s) \leq t$ and conversely \emph{search-time}$_G(s) = t$ implies \emph{search-width}$_G(t) \leq s$.

For a given time $t$, $\sigma$ is \emph{space-optimal} if $width(\sigma) =$ \emph{search-width$_G(t)$} with $length(\sigma) = t$. For a given number of search\-ers $s$, $\sigma$ is \emph{time-optimal} if $length(\sigma) =$ \emph{search-length$_G(s)$} with $width(\sigma) = s$. 

\section{Search-Time Lower Bound}

The lower bound for search time on a digraph does not come as easily as the lower bound for undirected graphs of $\big \lceil \frac{n-s}{s-1} \big \rceil + 1$ shown by Brandenburg and Herrmann \cite{Brandenburg2006UQ} since the reasoning used there does not apply to the directed case. That is, a search strategy on a digraph can leave a node unguarded without suffering from recontamination unlike in the undirected case. We follow a completely different avenue to the lower bound.

Given a search strategy $\sigma$ we can construct a set system $S = \{ S_1, \dots, S_t \}$ where each set corresponds to the placement of searchers in a single step of $\sigma$. Thus, $t$ represents the number of steps the strategy requires. We have the following conditions for such a set system to correspond to a valid and complete search strategy.

\begin{enumerate}
\item $ | S_i | \le s $
\item If $u, v$ are adjacent nodes in $G$ then there exists an $S_i$ where $u, v \in S_i$
\end{enumerate}

The first condition reflects the fact that we have $s$ searchers to work with while the second condition ensures that every edge in $G$ will be cleared.

As a result we have the following fact about $S$.

\begin{equation}
\forall i \hspace{1mm} \exists j \hspace{1.5mm} \text{such that} \hspace{1.5mm} S_i \cap S_j \neq \emptyset
\label{eqn:lb}
\end{equation}

Equation \ref{eqn:lb} comes from condition 2 and the fact that $G$ is connected since a set $S_i$ without an intersection with some other set would constitute a separate connected component violating our assumption of connectedness.

Note, a search strategy will also induce an ordering of $S$, $\Omega$, which dictates how the search strategy unfolds. Notice that every search strategy induces a unique set system while a given set system may correspond to several search strategies depending on the ordering. Next we define the \emph{progress} of a set which will be utilized in the lower bound proof.

\begin{defn}
The $progress$ of a set $S_i$ in an ordering $\Omega$ is the number of elements of $S_i$ which have not been seen in any previous set in $\Omega$.
\end{defn}

The progress of a set corresponds to the number of new nodes visited in that step of the corresponding search strategy.

Now, we present the search time lower bound on directed graphs which utilizes the set system notion.

\begin{thm}
For every connected digraph $G$ with $| G | = n$ and integer $s$ such that $s$ is at least the search number of $G$ all search strategies require at least $\big \lceil \frac{n-s}{s-1} \big \rceil + 1$ steps to clear $G$.
\end{thm}

\begin{proof}
Assume we are given an arbitrary search strategy $\sigma$ for $G$. First, we construct the corresponding set system $S$ for $\sigma$. Then, we construct a meta-graph on $S$ where each meta-node represents a set $S_i \in S$ and there is an \emph{undirected} edge between two meta-nodes if their corresponding sets have a non-empty intersection. Call the resulting graph $G_S$. Then, notice that equation (1) and our assumption of connectedness implies that $G_S$ is connected.

Now, we present a special ordering $\Omega'$ for $S$ by performing a depth-first search of $G_S$ initialized on any node of $G_S$. The order in which meta-nodes are visited in the DFS makes up $\Omega'$. This ordering may differ from that of $\sigma$ and is created purely for the proof of bounding the number of sets, $t$.

Then, we can bound the progress $\rho$ made by this ordering as follows. The first set visited in $\Omega'$ has a progress bounded above by $s$ from condition (1) and the fact that there are no previous sets in $\Omega'$. Then, every subsequent set $S_i$ in $\Omega'$ has a progress bounded above by $s-1$ since, by the DFS style ordering, there will be a set located earlier in $\Omega'$ which was connected to $S_i$, indicating a non-zero intersection. Thus, if there are $t$ sets, the total progress is bounded above by $s + (t-1)(s-1)$. Furthermore, $\rho$ is bounded below by $n$ as it is a necessary condition that every node in $G$ be visited by a searcher in order to clear all edges of $G$.

Thus, we have

\begin{align*}
n \le \rho &\le s + (t-1)(s-1) \\
n - s &\le (t-1)(s-1) \\
\frac{n-s}{s-1} &\le t - 1 \\
\end{align*}

Finally, since $t$ must be an integer we have

\begin{equation*}
t \ge \Big \lceil \frac{n-s}{s-1} + 1 \Big \rceil = \Big \lceil \frac{n-s}{s-1}\Big \rceil + 1
\end{equation*}

Therefore, we have shown that for an arbitrary search strategy, the corresponding set system requires at least $\big \lceil \frac{n-s}{s-1}\big \rceil + 1$ sets and thus any search strategy for $G$ must take at least this number of steps.
\end{proof}

In the next section we prove the hardness of computing the search time of a DAG.

\section{Hardness for DAGs}

In this section we prove that computing the search time of a DAG for a given number of searchers is NP-complete. Consider the GRAPH SEARCHING problem as determining the minimum number of steps required to clear an input directed graph $G$ on $n$ nodes with $s$ searchers. The decision version asks if $G$ can be cleared in $t$ steps. To do this we introduce two concepts required for the hardness proof: \emph{B-sections} and the \emph{loss} function. First, consider a branching node $v$ attached to $m$ directed paths $b_1, b_2, \dots, b_m$ all beginning at $v$ where $|b_i| \geq 1$. We refer to such structures as \emph{B-sections} and a sample $B$-section can be seen in Figure \ref{fig:hardness_example}. $B$-sections will be used in Section 5 for our decomposition theorem. Next, we define our \emph{loss} function. We know that an optimal strategy with $s$ searchers visits $s$ new nodes in the first step and $s-1$ new nodes in each subsequent step. All strategies can visit $s$ new nodes in the first step. Thus, a non-optimal searcher placement is one in which $s-1$ new nodes are not visited in a given step. Note, this excludes the final step where there may not be enough nodes left to visit $s-1$ new nodes. For this reason, an alternative definition is a placement in which two or more searchers are left stationary.

\begin{defn}
The loss function associated with a search strategy $\sigma$, denoted $loss(\sigma)$, is a count of the number of non-optimal searcher placements in $\sigma$.
\end{defn}

\begin{figure*}
	\centering
		\includegraphics[scale=0.5]{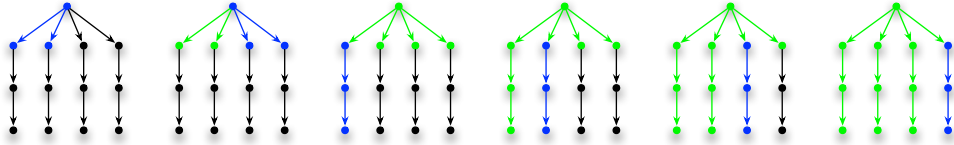}
	\caption{An example search strategy with zero loss}
	\label{fig:hardness_example}
\end{figure*}

Now, consider the $B$-section in Figure \ref{fig:hardness_example} and the search strategy shown which uses 3 searchers. The blue nodes represent the searcher placements of the current step while green nodes represent already visited nodes. The blue edges indicate edges which are being cleared in the current step while green edges are edges which have been cleared in a previous step. Notice how the search strategy partially clears each of the branches before finishing them off in a single step. The ability to avoid loss when moving to a new branch lay in the strategy's ability to ``set up'' the number of nodes left in each branch, after partially clearing the branches, as a multiple of $s-1$. This ensured that the clearance of each branch ended exactly at the leaf nodes and did not spill over into the next branch.

Now, we can generalize this idea to capture how a strategy would have to behave to ``set up'' the branches of a general $B$-section in a similar fashion in order to achieve zero loss. Consider a $B$-section with $m$ branches $b_1, \dots, b_m$ each of length $d_1, \dots, d_m$ where $d_i$ counts all nodes in $b_i$ other than the branching node. The question of whether or not zero loss can be achieved comes down to whether we can end the clearance of each branch exactly at the branch's final node. Therefore, we are asking whether we can move across the top of the $B$-section with the $s$ searchers such that after this initial sweep the number of nodes remaining to be cleared in each $b_i$ is a multiple of $s-1$. If this is the case, the branches could then be cleared one at a time with all $s$ searchers with the clearance ending exactly at the last node of each $b_i$ ensuring no loss when moving between branches.

The problem of the initial sweep across the top of the $B$-section can be phrased as an instance of a BIN PACKING variant. First, notice that there is a single value $0 \le x_i \le s-2$ for each branch that makes the number of nodes remaining in $b_i$ a multiple of $s-1$. Thus, we wish to know if we can pack the $x_i$ into bins of size $s-1$ such that each bin is exactly full. The solution to this problem tells us if the $B$-section can be cleared with zero loss. However, we know that not all $B$-sections can be cleared with zero loss and we actually want to know the minimum loss achievable. This leads to a variant of the optimization version of BIN PACKING which we wish to solve. In the standard BIN PACKING problem we wish to minimize the number of bins used. Our problem is asking to minimize the number of partially full bins. That is, we want to maximize the number of exactly full bins as each partially full bin represents a loss. We refer to this as the EXACT BIN PACKING problem and in the decision version denote the number of allowable partially filled bins by the parameter $p$. First, we show that the EXACT BIN PACKING problem remains strongly NP-hard even when $p$ is fixed to 0.

\begin{lem}
The EXACT BIN PACKING problem with $p=0$ is strongly NP-complete.
\label{lem:exact_bin_packing}
\end{lem}

\begin{proof}
Consider an instance of the decision version of BIN PACKING with items $X = \{x_1, \dots, x_n\}$, bin size $V$, and $b$ available bins. Then, let $r = Vb - \sum x_i$. Here, $r$ is the total remaining bin space (regardless of packing) for the BIN PACKING instance.

Now, we construct an instance of the decision version of the EXACT BIN PACKING problem with items $X' = X \cup 1_r$ where $1_n$ is a set containing $n$ 1's, bin size $V$, and $p=0$. Then, the EXACT BIN PACKING instance has a solution iff there is a solution to the instance of BIN PACKING.

If the BIN PACKING instance can be packed with $b$ bins then in the EXACT BIN PACKING instance we have exactly the required number of 1's to fill in the rest of the space leaving all exactly full bins, i.e. $\sum x_i + r = Vb$. However, if the BIN PACKING instance requires more than $b$ bins then there will not be sufficient 1's to fill in the space of additional bins and therefore there will be at least one partially filled bin, i.e. $\sum x_i + r = Vb < V(b+k)$ for some $k>0$.
\end{proof}

Then, it follows from the above result that the general version of EXACT BIN PACKING with arbitrary $p$ is also strongly NP-hard.

\begin{cor}
The EXACT BIN PACKING problem is strongly NP-complete.
\end{cor}

Now, we formally show the hardness of the graph searching problem on $B$-sections using the above result. First, we show that the GRAPH SEARCHING problem remains hard even for fixed $t = \big \lceil \frac{n-s}{s-1} \big \rceil + 1$ and restricted $G$.

Before we present the proof we will introduce some facts about the GRAPH SEARCHING problem. First, the search time of a strategy can be computed from the loss as $t = \big \lceil \frac{n-s + loss}{s-1} \big \rceil + 1$. Also, recall that the lower bound for clearing a graph is $t_{min} = \big \lceil \frac{n-s}{s-1} \big \rceil + 1$. Then, notice that $t_{min}$ can be achieved with a range of losses which depends on the values of $n$ and $s$. Namely, $t_{min}$ will be achieved by any strategy with $0 \leq loss \leq \big [ \big \lceil \frac{n-s}{s-1} \big \rceil - \frac{n-s}{s-1} \big ] (s-1)$. We refer to the upper bound by $loss_{max}$.

\begin{lem}
The GRAPH SEARCHING problem on B\hyp{}sections with $t = t_{min}$ is NP-complete.
\end{lem}

\begin{proof}
Consider an instance of the decision version of EXACT BIN PACKING with items $X = \{x_1, \dots, x_m\}$, bin size $V$, and $p = 1$. We construct an instance of GRAPH SEARCHING by transforming the $x_i$ into paths $\rho_i$ of length $x_i$ and attaching each $\rho_i$ to a distinguished branching node $\beta$. Additionally, we attach a path $b$ of length $\big \lceil \frac{\sum_i x_i- V}{V} \big \rceil V - (\sum_i x_i - V)$ to $\beta$. Call the resulting graph $G$ and notice that $G$ is a $B$-section. Let the GRAPH SEARCHING instance have $s=V+1$ and $t = \big \lceil \frac{|b| + \sum x_i- V}{V} \big \rceil + 1$. Then, the GRAPH SEARCHING instance has a solution iff there is a solution to the instance of EXACT BIN PACKING.

Notice that we have chosen the length of $b$ such that $loss_{max} = 0$, therefore $|b| + \sum x_i- V$ is a multiple of $V$. Also, the chosen $t = t_{min}$. Thus, if the EXACT BIN PACKING problem has a solution then $G$ can be cleared by first placing a searcher at $\beta$ and on every $\rho_i$ of a bin in a sweep across the top of $G$. The clearance of $b$ is included in the final step of the strategy. This search strategy has $loss = 0$ and will be able to clear the graph in $t_{min}$ steps.

In the other direction, given that $G$ can be cleared in $t = t_{min}$ steps we show how to obtain a solution to the EXACT BIN PACKING instance by progressively restricting how such a strategy must behave. Again, the structure of $G$ is such that $loss_{max} = 0$ so the strategy clearing $G$ cannot incur any loss. Thus, we can immediately rule out strategies which split into multiple groups; that is, any strategy in which the subgraph induced from searcher placements in a step does not form a connected component (ignoring the directions on edges) as such a step incurs a minimum loss of one. Then, to clear every branch we must leave a guard on $\beta$ as it is required to clear the first edge in each branch. Thus, since we cannot incur any loss, no node other than $\beta$ can be revisited in any step else the strategy would not visit $s-1$ new nodes. Therefore, we cannot partially clear any branch. Then, since each $\rho_i$ has length less than $V$, the strategy will fully clear some number of branches in every step of the strategy. Now, observe that we have restricted the allowable strategies such that they can only differ from the one described above by a re-ordering of steps. Thus, the $\rho_i$ cleared in each step are placed in a bin and the resulting bins make up the packing. Note, if $b$ had non-zero length we do not include it in the packing and thus get at most one partially full bin.
\end{proof}

From this we get our main result.

\begin{thm}
The GRAPH SEARCHING problem on $B$-sections is NP-complete.
\end{thm}

Furthermore, hardness on $B$-sections implies hardness on all its superclasses in the directed setting which includes directed trees, DAGs and all their directed superclasses. Therefore, we see an interesting comparison to computing the search number on undirected graphs where the problem becomes efficiently solvable when we move from general graphs to trees. However, computing the search time does not become efficiently solvable even when restricting the input graph to a $B$-section. In the following section, despite the hardness of the search time problem, we will introduce an efficient approximation algorithm for searching general digraphs.

\section{Our Search Algorithm}

\subsection{Searching Digraphs}

Since the graph searching problem is NP-hard even on B-sections, the task of clearing networks, which are general digraphs, is also NP-hard. We present a method for clearing a general digraph which works in two phases. We first compute a \emph{feedback vertex set (FVS)} for the network and place permanent guards at these nodes. Formally, an FVS is a set of nodes whose removal leaves a graph without cycles. Thus, by doing so, we are left with a DAG which can be cleared by our Plank algorithm given in the next subsection. The procedure for searching general digraphs is outlined in Algorithm \ref{alg:digraph}.

\begin{algorithm}
\caption{Search Digraph}
\begin{algorithmic}
\Require The input graph $G$ 
\Ensure A search strategy $\sigma = (V_1,\dots,V_t)$
\State Compute an FVS for $G$
\State Place permanent guards, $p$, on source nodes of FVS to create a DAG $G'$ that needs to be cleared
\State Run Plank algorithm on $G'$ to compute a search strategy $\sigma = (V_1,\dots,V_t)$
\State \Return $\sigma = (V_1 \cup p,\dots,V_t \cup p)$ 
\end{algorithmic}
\label{alg:digraph}
\end{algorithm}

Now, while the focus of our work is the search time, we can make some optimizations with respect to the FVS required to search a digraph in order to reduce the number of searchers used to clear $G$. First, we utilize a \emph{sliding FVS} which only places searchers on FVS nodes for as long as they are required. We say an FVS node is required when one of its neighbouring nodes is visited for the first time and is no longer required when all its neighbouring edges have been cleared. Thus, as the search strategy moves across the graph we have FVS nodes come online and then go offline when they are no longer required which reduces the total number of searchers used to clear the graph since at any time only a subset of the FVS will be active.

Furthermore, in the special case of social networks, we can leverage knowledge of the structure of the network to our advantage. Finding the minimal FVS is an NP-hard problem and thus we are reduced to using heuristic algorithms for the task. However, as most real world social networks exhibit a power-law degree distribution we know that the hub nodes will often be required in the FVS. This idea was utilized in \cite{Kang2011} to find communities in real world networks and we take a similar approach when considering social networks. In \cite{Kang2011} the \emph{k-hubset} is defined as the set of nodes with the top $k$ highest degrees. For our optimization, we compute the $k$-hubset of $G$ and take its union with the computed FVS to arrive at the set of permanent guards to be used for the sliding FVS. The mentality behind adding the $k$-hubset to the FVS is that any nodes in the $k$-hubset that are not included in the FVS will be visited many times during the clearance of $G$ due to their high connectivity and thus removing them preemptively will reduce the search time. Of course, there remain the questions of what value to choose for $k$ and how the size of the final FVS will be affected which we explore in our experiments. In the following section we present our Plank algorithm for searching a DAG.

\subsection{Plank Algorithm}

Our Plank algorithm works in a depth-first manner with some modifications specific to the graph searching problem. The name comes from a description of how searchers are placed in subsequent steps. Imagine a long plank of wood lying on the ground. We can move this plank by picking up one end until the plank is upright and then letting it drop in the direction we wish to travel. By repeatedly moving in this way we move the plank a distance equal to its length each time. Then, we can think of the plank as $s$ searchers placed adjacently on a graph so that moving the plank corresponds to visiting $s-1$ new nodes.

The Plank is a two-phase algorithm for computing its search strategy for a DAG, $G$. In a pre-processing step the algorithm computes an edge ordering for $G$, $\Psi$, and in the second pre-processing step it compiles a search strategy from $\Psi$. In Algorithm \ref{alg:mDFS} below, $mDFS$ refers to a modified depth-first search designed specifically for the Plank algorithm. Our $mDFS$ operates similarly to the $DFS$ algorithm, but with a special stopping condition: we backtrack if the current vertex has an unexplored incoming edge. This ensures we do not allow any recontamination from uncleared incoming edges as our strategy does not leave stationary guards at vertices. The Plank's high level execution proceeds as follows:

\begin{enumerate}
\item Run $mDFS$ on $G$ to produce an edge ordering $\Psi$
\item Convert $\Psi$ into a search strategy using $s$ searchers
\end{enumerate}

Now we present the Plank's subalgorithms. First, we have the pseudocode for the $mDFS$ algorithm in Algorithm \ref{alg:mDFS}. We assume all nodes in $G$ are initially labelled as unvisited and all edges as unexplored.

\begin{algorithm}
\caption{mDFS}
\begin{algorithmic}
\Require Input graph $G$ and the current node $v$
\Ensure An edge ordering $\Psi$
	\State $\Psi \gets []$
	\If {$v$ has no unexplored incoming edges}
		\State Label $v$ as visited
		\ForAll {edges $e$ in $G.outEdges(v)$}
			\If {edge $e$ is unexplored}
				\State $\Psi.append(e)$
				\State Label $e$ as explored
				\State $w \gets G.adjacentVertex(v, e)$
				\State $\Psi.append(mDFS(G, w))$ 
			\EndIf
		\EndFor
	\EndIf
	\State \Return $\Psi$
\end{algorithmic}
\label{alg:mDFS}
\end{algorithm}

Note, in the case that every node in $G$ is not visited in a call to $mDFS$, we continue re-calling the algorithm passing in an unexplored node until there are no more unexplored nodes in $G$. If there are multiple edge labelings, they are appended together to make a master edge labelling.

Next, we show how to convert the resulting edge labelling, $\Psi$, into a search strategy for $G$ using $s$ searchers (Algorithm \ref{alg:strat_construction}). In summary, $\Psi$ is traversed adding nodes to the current step in the search strategy until a step has reached $s$ placements. After $\Psi$ has been traversed we will have all the steps which make up the Plank search strategy $\sigma$. This procedure is captured in the pseudocode of Algorithm \ref{alg:strat_construction} where $V_c$ represents the nodes present in the current step.

\begin{algorithm}[h]
\caption{Construct Strategy}
\begin{algorithmic}
\Require Sequence $\Psi$ and the number of searchers $s$
\Ensure a search strategy $\sigma = (V_1,\dots,V_t)$
	\State $\sigma, V_c, cleared \gets \emptyset$
	\ForAll {edges $e$ in $\Psi$}
		\If {$e$ not in cleared}
			\If{$nodes(e)$ not in $V_c$}
				\State $V_c \gets V_c \cup nodes(e)$
				\State $cleared \gets cleared \cup e$
			\EndIf
		\EndIf
		\If {current step contains $s$ placements}
			\State update cleared with the edges cleared in the\\ \hspace{10mm}current step
			\State $\sigma.append(V_c)$
			\State $V_c = \emptyset$
		\EndIf
	\EndFor
	\State \Return $\sigma = (V_1,\dots,V_t)$
\end{algorithmic}
\label{alg:strat_construction}
\end{algorithm}

As it turns out, the strategy presented in Figure~\ref{fig:example_strat} is an example of a strategy produced by the Plank algorithm. For the $mDFS$ algorithm initialized at node 1 we get $\Psi =$ [(1,2),\hspace{1pt}(2,4),\hspace{1pt}(3,4),\hspace{1pt}(4,5),\hspace{1pt}(5,8),\hspace{1pt}(4,6),\hspace{1pt}(6,7),\hspace{1pt}(7,8),\hspace{1pt}(7,9)]. Then, the search strategy construction algorithm produces $\sigma = [(1,2,4,3), (4,5,8,6), (6,7,8,9)]$. Note, we can see that the strategy does not move passed node 4 while there are uncleared incoming edges. Similarly, if there were additional nodes below node 8, they would not be visited until the edge from node 7 to node 8 had been cleared.

\section{Analysis}

\subsection{Approximation Bounds}

In this section we will show that the Plank strategy is a $(2 + f_O)$-approximation algorithm for searching DAGs, where $f_O$ is an instance determined parameter, and motivate its performance on typical DAGs.

First we introduce some definitions to be used in the following proofs. We define four types of DAGs referred to as \emph{sections}. We have already seen the definition of $B$-sections in Section 4 (Fig.~\ref{fig:sections}(a)). Second are sections that resemble $B$-sections, except that the direction of each edge is reversed. That is, the structure is the same as a $B$-section, but with all branches directed towards a distinguished root which we refer to as \emph{R-sections} (Fig.~\ref{fig:sections}(b)). Next, we have sections which look like diamonds, or \emph{D-sections} (Fig.~\ref{fig:sections}(c)). These sections have a start node, two or more node disjoint branches, and an end node with branches originating at the start node and ending at the end node. Finally, we have simple directed paths, or \emph{P-sections} (Fig.~\ref{fig:sections}(d)). Note, the blue and red nodes mark the top and bottom nodes of a section respectively.

\begin{figure*}
	\centering
		\includegraphics[scale=0.6]{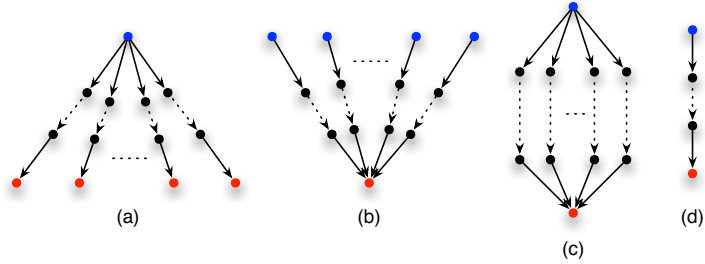}
	\caption{A sample (a) $B$-section (b) $R$-section (c) $D$-section and (d) $P$-section}
	\label{fig:sections}
\end{figure*}

We prove in section \ref{sec:dag_decomp} that any DAG can be decomposed into sections of the above four types and assume this holds for the remainder of the analysis. \\

To begin, we first prove an approximation bound for zero-overlap DAGs and then modify the bound to include the full range of DAGs. We define the \emph{overlap} of a node $v$ by

\[
	overlap(v) =
		\begin{cases}
			r	&	\text{if } v \text{ is a top/bottom node in} \geq 3 \text{ sections} \\
			0	&	\text{else}
		\end{cases}
\]

Where $r$ is the total number of sections for which $v$ is a top or bottom node.

Then, the overlap of a DAG $G$, denoted $\Omega$, is defined as $\Omega = \sum_{u \in V} overlap(u)$. A DAG is said to be a \emph{zero-overlap} DAG if $\Omega = 0$ and indicates a DAG in which each section overlaps with at most one other section. The following analysis assumes a zero-overlap DAG.

We bound the number of steps required by the Plank strategy by bounding the loss measure we introduced in Section 4. To that end, we consider the loss the Plank strategy can achieve when taking an arbitrary step in its clearance. We have three cases for how the strategy moves between steps: (1) the strategy remains within a single section, (2) the strategy finishes clearing a section and moves onto previously unvisited sections, or (3) the strategy finishes clearing a section and returns to a partially cleared section. We investigate these cases in three claims below.

\begin{claim}
A step taken by the Plank strategy described by Case 1 can incur a loss of no more than 2.
\end{claim}

\begin{proof}
Recall that the Plank strategy will move across branches of a section one at a time. Thus, when moving between branches in a $B$/$R$-section the strategy will revisit the top/bottom node of the section. Thus, the strategy will incur a loss if the previous branch was not cleared in a single step. On the other hand, when moving between branches of a $D$-section the Plank strategy will revisit both the top and bottom nodes incurring a loss of two if the previous branch was only partially cleared. Finally, a $P$-section trivially cannot incur a loss.
\end{proof}

\begin{claim}
A step taken by the Plank strategy described by Case 2 can incur a loss of no more than 2.
\end{claim}

\begin{proof}
When moving to a new section, besides the nodes connecting sections, every node is being visited for the first time. Thus, we again have a worst case loss of 2, in the situation where we finish clearing a $D$-section $\theta$ and move onto clearing a section which overlaps with the top node of $\theta$. In contrast, moving to a downstream section can only incur a worst case loss of 1, when the bottom node of $\theta$ is revisited, as every other node is visited for the first time.
\end{proof}

\begin{claim}
A step taken by the Plank strategy described by Case 3 can incur a loss of no more than $\lceil \frac{s}{2} -1 \rceil$.
\end{claim}

\begin{proof}
When returning to a $B$-section $\theta_B$, we incur a loss of 1 by returning to the branching node. Then, consider the case where $\theta_B$ is entirely cleared with the available searchers and the strategy must again move to a new section or return to another $B$-section. Here, moving to a new section would incur no extra loss as the new section would be downstream from $\theta_B$. However, the strategy could continue clearing $B$-sections and returning up to more partially cleared $B$-sections incurring a loss each time this occurs. The number of times the strategy could return to a $B$-section is bounded by the number of searchers available, $s$, and the minimum size of the portion of the $B$-section left to be cleared, $2$. Thus, the Plank strategy could take a single step which incurs a loss of $\lceil \frac{s}{2} -1 \rceil$ as the sections which are revisited must have at least one node other than the branching node not yet visited. An analogous situation occurs when returning to an $R$-section. Note, $D$-sections cannot be partially cleared and thus do not come up in Case 3 steps.
\end{proof}

Now, we can divide an arbitrary Plank strategy into steps adhering to Case 1, 2, or 3. Thus, w.l.g. we can investigate the approximation ratios for steps of each type to arrive at an overall approximation ratio. We group Case 1 and 2 steps together as \emph{Type 1} steps while Case 3 steps are referred to as \emph{Type 2} steps.

\begin{lem}
The Type 1 steps have an approximation ratio of no more than $2$.
\label{lem:approx1}
\end{lem}

\begin{proof}
Given $s$ searchers, consider $k$ steps incurring a loss of 2. Then, the total number of nodes from Type 1 steps, $n$, is at least $k(s+1)$ for $s=4$ and $(k-1)(s-3) + s + 2$ for $s \ge 5$. The expression for $s=4$ comes from the fact that a $D$-section being cleared with $4$ searchers cannot enter into a pattern which incurs a loss of $2$ for multiple branches, instead they can only incur a loss of $2$ when clearing the second branch of a $D$-section. The case of $s=4$ is captured in Fig. \ref{fig:approx} (a). For the $s \ge 5$ expression, we visit $s-3$ new nodes in each step except the last step where we may run out of nodes left to visit in which case $2$ additional nodes is a minimum. Notice, the additional $s$ comes from the fact that all strategies visit $s$ nodes in the first step and incur no loss. A simple example for $s=5$ is presented in Fig. \ref{fig:approx} (b).

\begin{figure}
	\centering
		\includegraphics[scale=0.7]{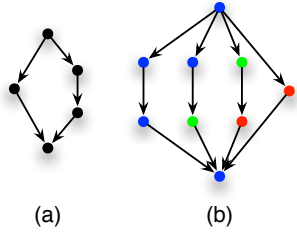}
	\caption{The (a) $s=4$ case and (b) an example for $s=5$. Notice, in (a) any additional branches in the $D$-section would not incur a loss of $2$ given $4$ searchers. In (b), we have $s=5$ and $k=2$. The blue nodes are the initial $s$ nodes which any strategy can visit without loss in the first step which sets up subsequent steps incurring a loss of $2$. The green nodes are the $s-3$ new nodes visited in each step while the red nodes show how the last step only requires $2$ additional nodes to incur a loss of $2$.}
	\label{fig:approx}
\end{figure}

Then, the loss is bounded by $2k$ or $\frac{2n}{5}$ for $s=4$ and $\frac{2(n-5)}{s-3}$ for $s \ge 5$. Now, we can compute an approximation ratio by comparing the lower bound $\lceil \frac{n-s}{s-1} \rceil + 1$ to the expression $\lceil \frac{n-s + loss}{s-1} \rceil + 1$. First, in the case for $s=4$ we have,

\begin{equation}
\Big \lceil \frac{n-s + \frac{2n}{5}}{s-1} \Big \rceil + 1 \le \frac{n-s + \frac{2n}{5}}{s-1} + 2 = \frac{7n + 5s - 10}{5(s-1)}
\end{equation}

\noindent And

\begin{equation}
\Big \lceil \frac{n-s}{s-1} \Big \rceil + 1 \ge \frac{n-s}{s-1} + 1 = \frac{n-1}{s-1}
\end{equation}

\noindent Then

\begin{equation} \label{eq:bound1}
\frac{\big \lceil \frac{n-s + \frac{2n}{5}}{s-1} \big \rceil + 1}{\big \lceil \frac{n-s}{s-1} \big \rceil + 1} \le \frac{\frac{7n + 5s - 10}{5(s-1)}}{\frac{n-1}{s-1}} = \frac{7n + 5s - 10}{5(n-1)}
\end{equation}

\noindent Where \eqref{eq:bound1} is bounded above by $2$ for $n \ge 7$. Then, it is easy to verify by hand that for all DAGs with $5$ or $6$ nodes the Plank strategy requires no more than $3$ steps while the lower bound requires $2$ steps.

Second, in the case for $s \ge 5$ we have,

\begin{equation}
\Big \lceil \frac{n-s + \frac{2(n-5)}{s-3}}{s-1} \Big \rceil + 1 \le \frac{ns - n + s^2 -5s - 4}{(s-3)(s-1)}
\end{equation}

\noindent Then

\begin{equation} \label{eq:bound2}
\frac{\big \lceil \frac{n-s + \frac{2(n-5)}{s-3}}{s-1} \big \rceil + 1}{\big \lceil \frac{n-s}{s-1} \big \rceil + 1} \le \frac{\frac{ns - n + s^2 -5s - 4}{(s-3)(s-1)}}{\frac{n-1}{s-1}} = \frac{ns - n + s^2 -5s - 4}{(n-1)(s-3)}
\end{equation}

\noindent Where \eqref{eq:bound2} is bounded above by $2$ for $s \le \frac{n+3}{2}$. Note, we only consider the case where $s \le \frac{n+3}{2}$ since when $s > \frac{n+3}{2}$ all $n$ nodes will be cleared in $2$ steps as no more than $3$ nodes will remain stationary between steps.

Therefore, an arbitrary number of Type 1 steps has an approximation ratio of no more than $2$.
\end{proof}

\begin{lem}
The Type 2 steps have an approximation ratio of no more than $2$.
\label{lem:approx2}
\end{lem}

\begin{proof}
Given $s$ searchers, consider $k$ steps incurring a loss of $\lceil \frac{s}{2} -1 \rceil$. Then, notice that the upper bound on the number of steps required by a strategy on zero-overlap DAGs is $n - s + 1$ as we visit at least one new node in each step. Thus, $k \le n - s + 1$ giving a loss bounded above by $k \lceil \frac{s}{2} -1 \rceil \le \frac{ns - s^2 + s}{2}$ since $\lceil \frac{s}{2} -1 \rceil \le \frac{s}{2}$. Now, notice that $\frac{ns - s^2 + s}{2} \le \frac{n}{2}$ for $s \le n$ which holds for all search strategies. Therefore, we can compute the approximation ratio as,

\begin{equation}
\Big \lceil \frac{n-s + \frac{n}{2}}{s-1} \Big \rceil + 1 \le \frac{n-s + \frac{n}{2}}{s-1} + 2 = \frac{3n + 2s - 4}{2(s-1)}
\end{equation}

\noindent Then

\begin{equation} \label{eq:bound3}
\frac{\big \lceil \frac{n-s + \frac{n}{2}}{s-1} \big \rceil + 1}{\big \lceil \frac{n-s}{s-1} \big \rceil + 1} \le \frac{\frac{3n + 2s - 4}{2(s-1)}}{\frac{n-1}{s-1}} = \frac{3n + 2s - 4}{2(n-1)}
\end{equation}

\noindent Where \eqref{eq:bound3} is bounded above by $2$ for $s \le \frac{n-1}{2}$. Note, we only consider the case where $s \le \frac{n-1}{2}$ since when $s > \frac{n-1}{2}$ the number of nodes remaining after the first step is less than $\frac{n}{2}$ and therefore the loss cannot exceed this value.

Therefore, an arbitrary number of Type 2 steps has an approximation ratio of no more than $2$.
\end{proof}

Thus, we get the following approximation bounds for the Plank strategy on zero-overlap DAGs.

\begin{lem}
The Plank algorithm is a $2$-approximation algorithm for computing the search time of a zero-overlap DAG.
\end{lem}

\begin{proof}
We consider an arbitrary instance of a Plank strategy. The steps of the strategy are all of Type 1 or 2. Then, the proof follows directly from Lemma's \ref{lem:approx1} and \ref{lem:approx2}.
\end{proof}

In practice, the number of searchers will often be much less than the size of the DAG, $s \ll n$, in which case $\eqref{eq:bound1} \approx \frac{7}{5} + O(\frac{s}{n})$, $\eqref{eq:bound2} \approx 1 + O(\frac{2}{s}) + O(\frac{s}{n})$, and $\eqref{eq:bound3} \approx \frac{3}{2} + O(\frac{s}{n})$. Furthermore, the structure of a DAG required to produce an approximation ratio for \eqref{eq:bound3} of $\frac{3}{2} + O(\frac{s}{n})$ is extremely artificial and would not show up in a large fraction of DAGs and \eqref{eq:bound1} only applies when $s=4$. In general, we expect the approximation ratio to closely resemble $1 + O(\frac{s}{n})$. Therefore, proving the usefulness of the Plank algorithm for typical zero-overlap DAGs. \\

Now, we must modify the bound for DAGs with nonzero overlap. The overlap of a DAG can be viewed as rough estimation of the density of the digraph. As such, DAGs move progressively towards resembling directed complete bipartite graphs (with all edges directed from one partition to the other) as the overlap increases. We take a conservative route and add to the bound of $2$ for zero-overlap DAGs an \emph{overlap factor}, $f_o$. The $f_o$ factor upper bounds the number of steps required to clear the number of possible edges incident on the overlapping nodes. It is defined as,

\begin{equation}
f_o = \Big ( \frac{\Omega}{n-1} \Big )
\end{equation}

and can often be approximated by $\frac{m}{n}$. Thus, combining the possible loss in zero-overlap DAGs and the potential loss in DAGs with overlap yields an approximation ratio that holds for all DAGs of $2 + f_o$.

\begin{thm}
The Plank algorithm is a $(2+f_o)$\hyp{}approximation algorithm for computing the search time of a DAG.
\end{thm}

Note, as we provide a lower bound on the length of a search strategy that is independent of the structure of the input DAG, our $f_o$ factor may take on large values for highly overlapping DAGs when the length of the Plank strategy, in reality, may not be far off the instance-optimal solution.

\subsection{Comparison to Splitting Strategies}

Another natural candidate for graph searching would be a BFS style strategy which we investigate next. We show that the DFS style strategy, our Plank algorithm, outperforms the BFS style strategies on a broad class of DAGs. We refer to BFS style strategies as \emph{splitting strategies} and define them as follows.

\begin{defn}
A splitting strategy is a search strategy which sends at least two searchers down as many branches of a section as possible.
\end{defn}

The way in which a splitting strategy distributes the search\-ers over the branches is arbitrary, but the key point is that such a strategy tries to split as much as possible, mimicking a BFS. As with the Plank algorithm, splitting strategies do not move passed nodes with unexplored incoming edges to avoid recontamination. Alternatively, we can think of splitting strategies as split and conquer style strategies.

\begin{lem}
The Plank strategy outperforms all splitting strategies in clearing a $B$-sections with any number of searchers.
\label{lem:b_sec}
\end{lem}

\begin{proof}
First, notice that if there are enough searchers available to clear a branch in a single step then both strategies are identical.  This would only occur if the splitting strategy had enough searchers to clear an entire branch as it allocates strictly less searchers per branch compared to the Plank strategy. Thus, we can restrict our attention to branch sets where every branch requires two or more steps to clear for all searcher distributions.

We consider the loss incurred on a $B$-section. A splitting strategy which distributes the available searchers among the $m$ branches will incur a loss of $m-1$ for each step it takes to clear the branches since $m$ searchers remain stationary. Then, by our assumption that each branch takes at least two steps to clear we see that a splitting strategy incurs a loss greater than or equal to $m-1$. Conversely, even if every branch required more than two steps to clear, the Plank strategy will never incur a loss greater than $m-1$. This follows from the fact that once a leaf node is reached the strategy will return to the branching node during a step that will possibly only visit $s-2$ new nodes incurring a loss of 1. If, however, the step ends exactly at the leaf node there will be no loss incurred. Since there are $m$ branches we will encounter this "($s-2$)-visiting step" a maximum of $m-1$ times giving a total loss of no greater than $m-1$.
\end{proof}

Furthermore, the proof for $R$-sections unfolds exactly as the proof for $B$-sections does with the worst case loss being less that or equal to the best case performance of any splitting strategy.

\begin{lem}
The Plank strategy outperforms all splitting strategies in clearing $R$-sections with any number of searchers.
\label{lem:r_sec}
\end{lem}

Next, we prove that the Plank strategy is optimal for $D$-sections. We refer to the start node by $n_s$ and the end node by $n_e$. Furthermore, we consider $D$-sections to have $m$ branches $b_1, \dots, b_m$ which each have endpoints $n_s$ and $n_e$ and no other nodes in common.

\begin{lem}
The Plank strategy outperforms all splitting strategies in clearing $D$-sections with any number of searchers.
\label{lem:d_sec}
\end{lem}

\begin{proof}
As before, we will be considering splitting strategies with an arbitrary searcher distribution over branches, but will ignore branches which would be cleared in a single step in the splitting strategy as this would be mimicked exactly by the Plank strategy, i.e. these splitting strategies are indistinguishable from the Plank strategy.

Now, as we saw in Lemma~\ref{lem:b_sec}, the worst case behaviour of the Plank strategy on $B$-sections had a loss of $m-1$. For $D$-sections, the Plank strategy has a worst case loss of $2(m-1)$ because in addition to the top node we also revisit the bottom node $n_e$ when clearing each branch. Thus, both the bottom and top nodes are revisited when clearing subsequent branches adding 2 to the loss function each time. So, in a similar fashion to the proof for $B$-sections, a best case splitting strategy which requires 3 steps to clear each branch will require greater than or equal to the number of steps required by a Plank strategy. It remains to be shown that splitting strategies in which some number of branches require 2 steps to clear are no better than the Plank strategy. We investigate them each separately next.

We define the following three cases for splitting strategies: (1) each branch requires 2 steps to clear, (2) branches require 2 or 3 steps to clear, and (3) there exists a branch which requires greater than 3 steps to clear. We define the cases in this way because we will see that case 3 can be reduced to case 2. \\

\textbf{Case 1.} \emph{Branches each require 2 steps to clear.} \\
Here, the splitting strategy sends $s_i$ searchers down $b_i$ and the total number of searchers is $s = \sum_{i=1}^m s_i$. Again, considering a best case scenario for the splitting strategy, the clearance will require exactly two steps which leads to $2s_i - 1$ nodes in each branch. Now, the Plank strategy will send all $s$ searchers down some $b_i$ and then move onto the remaining branches. Sending all $s$ searchers down $b_i$ leads to two possible behaviours. If $2s_i - 1 > s$ then the Plank strategy cannot clear the entire $b_i$ branch in a single step and will have some nodes from $b_i$ left over to clear after step one. On the other hand, if $2s_i - 1 \le s$ then the Plank strategy will clear all of $b_i$ in a single step and have excess searchers available to start clearing other branches in the first step. Without loss of generality, we will assume $n_s$ and $n_e$ are a part of the first branch cleared in all subsequent cases.

First, in the case where $2s_i - 1 < s$ we partially clear $b_i$ and the number of nodes left to clear in $b_i$ after the first step is given by,

\begin{align*}
2s_i - 1 - s & = 2s_i - (s_1 + \dots + s_i + \dots + s_m) - 1 \\
& = s_i - (s_1 + \dots + s_{i-1} + s_{i+1} + \dots + s_m) - 1
\end{align*}

Then, in the second step we must clear all remaining branches as well as leave two searchers stationary. One searcher must be left at $n_s$ and the other at the furthest node reached in the partial clearance of $b_i$. The number of nodes left to clear is given by $\sum_{j \ne i} (2s_j - 1) + 2 + s_i - \sum_{j \ne i} s_j - 1 = s_i + \sum_{j \ne i} s_j + 2 - m = s + 2 - m $. Then, since $m \ge 2$, we have $s$ or fewer nodes left to clear in the second step with our $s$ available searchers.

Second, in the case where $2s_i - 1 \le s$ we clear $b_i$ in step one and clear an additional number of nodes with the excess searchers given by

\begin{align*}
s - (2s_i - 1) & = (s_1 + \dots + s_i + \dots + s_m) - 2s_i + 1 \\
& = (s_1 + \dots + s_{i-1} + s_{i+1} + \dots + s_m) - s_i + 1
\end{align*}

These excess searchers can clear any other branches since both $n_s$ and $n_e$ are guarded. Notice that some branch $b_r$ will be partially cleared at the end of step one using the excess searchers. Then, in the second step we must clear all remaining branches as well as leave three searchers stationary. One searcher must be left at $n_s$, one at $n_e$ and the other at the furthest node reached in the partial clearance of $b_r$. The number of nodes left to clear is given by $\sum_{j \ne i} (2s_j - 1) - (\sum_{j \ne i} s_j - s_i + 1) + 3 = s + 3 - m$. Here we have enough searchers for $m \ge 3$, but must investigate $m=2$ individually. In the case of $m=2$ where we only have $b_1$ and $b_2$ we are only required to leave 2 searchers stationary, namely one at $n_e$ and the other at the furthest node reached in the partial clearance of $b_2$. Thus, the number of nodes left to clear in step two is indeed $s + 2 - m$ which is achievable with the $s$ available searchers.

This shows that the Plank strategy matches the best case splitting strategy for case 1. \\

\textbf{Case 2.} \emph{Branches each require 2 or 3 steps to clear.} \\
Without loss of generality we let $b_1$ be a branch requiring two steps to clear, $b_2$ be a branch requiring three steps to clear and every other branch requiring two or three steps to clear. Thus, $b_1$ has $2s_1 - 1$ nodes, $b_2$ has $3s_2 - 2$ nodes, and every other branch has no more than $3s_i - 2$ nodes for $i \ne 1,2$. Also w.l.g. we choose to clear $s_1$ first in the Plank strategy. As in case 1, we consider the two scenarios where $2s_1 - 1 > s$ or $2s_1 - 1 \le s$.

First, in the case where $2s_1 - 1 > s$ we partially clear $b_1$ and the number of nodes left to clear in $b_1$ after the first step is given by,

\begin{align*}
2s_1 - 1 - s & = 2s_1 - (s_1 + \dots + s_i + \dots + s_m) - 1 \\
& = s_1 - (s_2 + \dots + s_m) - 1
\end{align*}

The rest of $b_1$ is cleared in step two with the number of excess searchers available after clearing $b_1$ given by $s - (s_1 - \sum_{j \ne 1} s_j - 1) - 1 = (s_1 + \dots + s_m) - s_1 + \sum_{j \ne 1} s_j = 2\sum_{j \ne 1} s_j$. Then, these excess searchers are used to partially clear the remaining branches. Note that $| b_i | \le 3s_i - 2$ and so we analyze a worst case scenario for the Plank strategy where each $b_i$ has all $3s_i - 2$ nodes. The number of remaining nodes to be cleared in step three is given by $\sum_{j \ne 1} (3s_j - 2) + 1 - 2\sum_{j \ne 1} s_j = \sum_{j \ne 1} s_j - 2m + 3$. Then, between steps two and three we must leave three searchers stationary. Thus, the number of nodes left to clear in step three is $\sum_{j \ne 1} s_j - 2m + 6$. We know that $s_1 \ge 2$ and $m \ge 2$ thus we have that $\sum_{j \ne 1} s_j - 2m + 6$ is less than or equal to $s$ allowing the Plank strategy to successfully complete the clearance in three steps.

Second, in the case where $2s_i - 1 \le s$ we clear $b_1$ in step one and clear an additional number of nodes with the excess searchers given by

\begin{align*}
s - (2s_1 - 1) & = (s_1 + \dots + s_m) - (2s_1 - 1) \\
& = (s_2 + \dots + s_m) - s_1 + 1
\end{align*}

These excess searchers can clear any other branches since both $n_s$ and $n_e$ are guarded. We again analyze a worst case scenario for the Plank strategy where each $b_i$ has all $3s_i - 2$ nodes. The number of nodes remaining after step one is given by $\sum_{j \ne 1} (3s_j - 2)  - (\sum_{j \ne 1} s_j - s_1 + 1) = 2\sum_{j \ne 1} s_j - 2m + 1 + s_1$. Here, between steps one and two we must leave searchers at $n_s$, $n_e$, and the last node reached in the partial clearance of some branch $b_i$. The number of nodes left to clear after step two is given by $2\sum_{j \ne 1} s_j - 2m + 1 + s_1 + 3 - s = \sum_{j \ne 1} s_j + 4 - 2m$. Again, we know that $m \ge 2$ and thus $\sum_{j \ne 1} s_j + 4 - 2m \le \sum_{j \ne 1} s_j < s$. Therefore, the Plank strategy can successfully finish the clearance in three steps.

This shows that the Plank strategy matches the best case splitting strategy for case 2. \\

\textbf{Case 3.} \emph{There exists a branch which requires greater than 3 steps to clear.} \\
The final case can be shown to reduce to Case 2. Consider $m$ branches which each require $t_i$ steps to clear where each $t_i \ge 2$. Now, as we have seen, if some branch $b_i$ requires $t_i$ steps to clear and some other branch $b_j$ requires $t_i + 1$ steps to clear, then excess searchers available from $b_i$ in step $t_i + 1$ will be inconsequential as $b_j$ already had enough searchers to clear $b_j$ by step $t_i + 1$. However, if instead, $b_j$ required $t_i + 2$ or greater steps to clear, the excess searchers from $b_i$ can actually have an impact on the number of steps required to clear $b_j$. In the best case, the excess searchers from $b_i$ allow $b_j$ to be cleared in only one additional step. However, since we know that each branch requires greater than or equal to two steps to clear, the best case for splitting strategies is to reduce branches for which $t_j \ge 4$ to requiring three steps leaving us in a situation resembling case 2. Note that this best case may not even be achievable given the structure of the $D$-section in question. \\

Thus, we have shown that the Plank strategy matches or outperforms all possible splitting strategies on an arbitrary $D$-section.
\end{proof}

Now, we have that the Plank strategy outperforms all splitting strategies on each of the sections individually. Then, the fact that any DAG can be decomposed into sections of our four types, which we prove in the next section, allows us to observe that the loss due to the Plank algorithm will be the same in its clearance of decomposed sections within a DAG as if they were being cleared in isolation conditioned on the length of the section's branches. For $R$-sections and $D$-sections the Plank algorithm will not move passed the bottom branching node and will thus return to one of the top nodes of the section (possibly after clearing sections above the current one) and ultimately clear the section with the same loss as if it was isolated. For $B$-sections the Plank strategy may clear downstream sections before returning to the branching node at the top of the section. Therefore, in order to be able to make a piecewise analysis of the DAG we require that the $B$, $R$, and $D$-sections contain branches of length $s$ or greater. While the analysis does not require this size restriction for individual sections, when analyzing a DAG without these ``large" sections there exist instances where a splitting strategy will incur no loss in some section where the Plank strategy does incur some loss due to the DFS nature of the Plank strategy. Thus we have the following result.

\begin{thm}
The Plank strategy outperforms all splitting strategies in clearing DAGs with ``large" $B$, $R$, and $D$-sections with any number of searchers.
\end{thm}

\begin{proof}
Two arbitrary sections in $G$, $\theta_1$ \& $\theta_2$ may be connected in three ways: (1) a bottom node of $\theta_1$ overlaps with a top node of $\theta_2$, (2) a bottom node of $\theta_1$ overlaps with a bottom node of $\theta_2$, or (3) a top node of $\theta_1$ overlaps with a top node of $\theta_2$.

In (1), the loss attributed to the overlapping node is divided between the sections. While $\theta_1$ is being cleared the loss is associated with $\theta_1$. Then, once $\theta_1$ becomes cleared and a strategy moves on to $\theta_2$, the loss will become associated with $\theta_2$. Thus, we see that the first step in which nodes from $\theta_2$ are cleared will not incur any loss, as is the case when clearing isolated sections.

In (2) we say $\theta_2$ is \emph{lateral} to $\theta_1$ and vice versa. In the Plank strategy, the overlapping node will be visited for the first time in the clearance of one of the sections. Then, when the node is reached in clearing the other section there will be a loss incurred. However, notice that this loss will also be incurred for splitting strategies. First, it is possible the splitting strategy reaches the overlapping node at the same time if $\theta_1$ and $\theta_2$ are being cleared simultaneously. However, in this case, the splitting strategy will incur a loss from all but one branch between both $\theta_1$ and $\theta_2$ since they are being cleared simultaneously and thus the extra loss incurred by the Plank strategy will also be incurred in the splitting strategy. On the other hand, if the splitting strategy does not reach the overlapping node in the same step, it too will incur an extra loss when the overlapping node is reached for the second time.

Finally, (3) mirrors the situations which arise in (2) and we see that the extra loss incurred by the Plank strategy is also incurred by the splitting strategy.

Then, the proof follows directly from Lemma's \ref{lem:b_sec}, \ref{lem:r_sec}, and \ref{lem:d_sec} and the fact that we can decompose any DAG into $B$, $R$, $D$, and $P$-sections.
\end{proof}

%For the sake of brevity, we give a high level description of our results regarding splitting strategies. We begin by proving the Plank strategy outperforms all splitting strategies on each type of section individually. The proof for $B$ and $R$-sections follows from a direct comparison of the loss required by any splitting strategy to the upper bound on the loss possible by the Plank strategy. Then, in the case of $D$-sections we must investigate several cases from which we again show the Plank strategy outperforms any splitting strategy.

%Now, the fact that any DAG can be decomposed into sections of our four types, which we prove in the next section, allows us to observe that the loss due to the Plank algorithm will be the same in its clearance of decomposed sections within a DAG as if they were being cleared in isolation conditioned on the length of the section's branches. The condition required is that the $B$, $R$, and $D$-sections contain branches of length $s$ or greater which we refer to as having \emph{large} sections. Thus we are able to show the following result in the full version of the paper \cite{simpson2014}.

%\begin{thm}
%The Plank strategy outperforms all splitting strategies in clearing DAGs with large $B$, $R$, and $D$-sections with any number of searchers.
%\end{thm}

\subsection{Decomposing a DAG}
\label{sec:dag_decomp}

We claim that a DAG can be decomposed into sections of our four types. Formally, we define a \emph{valid} decomposition as follows.

\begin{defn}
Given a DAG $G$ a decomposition $\Delta$ is $valid$ if and only if it consists of sections $\theta_i = (V_i, E_i)$ of type $B$, $R$, $D$, or $P$ such that $\bigcup_{i} V_i = V$, $\bigcup_{i} E_i = E$ and $E_i \cap E_j = \emptyset$ for all $i,j$. Additionally, sections may only overlap on top and bottom nodes.
\end{defn}

Next, we define an ordering among valid decompositions.

\begin{defn}
We say $\Delta_1 < \Delta_2$ if $\Delta_1, \Delta_2$ are valid decompositions and $\Delta_1$ can be obtained from $\Delta_2$ by some number of merge operations.
\end{defn}

A merge operation takes two valid sections and combines them to form a new valid section. Formally, given two sections $\theta_1 = (V_1, E_1)$ and $\theta_2 = (V_2, E_2)$, $merge(\theta_1, \theta_2) = (V_1 \cup V_2, E_1 \cup E_2)$. We outline the possible merge operations in the below table.

\begin{center}
  \begin{tabular}{ c | c }
    Components & Possible Merge Outcome \\ \hline
    $P$ & $B$, $R$, $D$, $P$ \\
    $B$ & $B$ \\
    $R$ & $R$ \\
    $D$ & $D$ \\
    $P$, $B$ & $B$ \\
    $P$, $R$ & $R$ \\
    $P$, $D$ & $D$ \\
    $B$, $R$ & $D$ \\
    $P$, $B$, $R$ & $D$ \\
  \end{tabular}
\end{center}

Then, we can define a minimality property for decompositions.

\begin{defn}
A decomposition $\Delta$ is minimal if $\neg \exists \: \Delta'$ such that $\Delta' < \Delta$.
\end{defn}

Finally, we show how to compute a minimal decomposition for any DAG. Consider the following procedure on a topological ordering $\Gamma$ of a DAG $G$. In the first phase we will move through $\Gamma$ one node at a time. Starting at the current node $v$ we will traverse $\Gamma$ for each outgoing edge of $v$ until we reach a node with multiple incoming edges or zero or multiple outgoing edges. This sequence of nodes will be appended to a list $\eta$. Phase one is presented in the pseudocode of Algorithm \ref{alg:decomp1}.

\begin{algorithm}[h]
\caption{Phase one of the minimal decomposition algorithm}
\begin{algorithmic}
\Require the topological ordering $\Gamma$
\Ensure the list $\eta$
	\State $\eta \gets \emptyset$
	\ForAll {nodes $v \in \Gamma$}
		\ForAll {outgoing edges $e$ of $v$}
			\State $u \gets e.destination$
			\State $seq \gets \{v, u\}$
			\While {$u$ has exactly one outgoing edge $e_o$}
				\State $u \gets e_o.destination$
				\State append $u$ to $seq$
			\EndWhile
			\State append $seq$ to $\eta$ 
		\EndFor
	\EndFor
\end{algorithmic}
\label{alg:decomp1}
\end{algorithm}

After this phase, each edge will be in a unique sequence in $\eta$. In a second phase, for each sequence $\lambda$ in our list we will combine $\lambda$ with other sequences located after $\lambda$ in $\eta$ which have not already been designated to a section to create a section of one of the four types. Once $\eta$ has been traversed each edge of $G$ will be in a unique section.

The way in which we combine sequences is as follows. Consider two sequences $\lambda_1 , \lambda_2$ made up of nodes $u_1, \dots, u_{m_1}$ and $v_1, \dots, v_{m_2}$ respectively. We proceed through a series of possible scenarios. First, if $u_1 = v_1$ and $u_{m_1} = v_{m_2}$ we combine $\lambda_1$ and $\lambda_2$ into a $D$-section. Second, if $u_1 = v_1$ we combine $\lambda_1$ and $\lambda_2$ into a $B$-section. Third, if $u_{m_1} = v_{m_2}$ we combine $\lambda_1$ and $\lambda_2$ into a $R$-section. Finally, if the previous three scenarios fail to be met, we leave $\lambda_1$ as a $P$-section. Phase two is captured in the pseudocode of Algorithm \ref{alg:decomp2}. Note that we refer to the first and last nodes in a sequence $\lambda$ by $\lambda_s$ and $\lambda_e$ respectively.

\begin{algorithm}[h]
\caption{Phase two of the minimal decomposition algorithm}
\begin{algorithmic}
\Require the list $\eta$
\Ensure a collection of sections of type $B$, $R$, $D$, and $D$
	\ForAll {sequences $\lambda \in \eta$}
		\State collect all unclaimed sequences $\alpha \in \eta$ such that \\ \hspace{4mm} $\lambda_s = \alpha_s$ in a list $L_1$
		\If {$L_1 \ne \emptyset$}
			\State collect all unclaimed sequences $\beta \in L_1$ \\ \hspace{9mm} such that $\lambda_e = \beta_e$ in a list $L_2$
			\If {$L_2 \ne \emptyset$}
				\State create a $D$-section from $L_2$ and $\lambda$
				\State mark $L_2$ and $\lambda$ as claimed
			\Else
				\State create a $B$-section from $L_1$ and $\lambda$
				\State mark $L_1$ and $\lambda$ as claimed
			\EndIf
			\State \textbf{continue}
		\EndIf
		\State collect all unclaimed sequences $\gamma \in \eta$ such that $\lambda_e = \gamma_e$ \\ \hspace{4mm} in a list $L_3$
		\If {$L_3 \ne \emptyset$}
			\State create a $R$-section from $L_3$ and $\lambda$
			\State mark $L_3$ and $\lambda$ as claimed
			\State \textbf{continue}
		\EndIf
		\State create a $P$-section from $\lambda$
		\State mark $\lambda$ as claimed
	\EndFor
\end{algorithmic}
\label{alg:decomp2}
\end{algorithm}

\begin{thm}
For any DAG $G$, Algorithm~\ref{alg:decomp1} and Algorithm~\ref{alg:decomp2} produce a minimal decomposition $\Delta$.
\end{thm}

\begin{proof}
Suppose there exists a decomposition $\Delta' < \Delta$. Thus, there exists two or more sections in $\Delta$ which can be merged. Without loss of generality, suppose there are only two sections $\theta_1, \theta_2$ which can be merged. Consider the sequences $\lambda_1, \dots, \lambda_n$ and $\mu_1, \dots, \mu_n$ that were combined to make $\theta_1$ and $\theta_2$ respectively. Then, it is easy to see that regardless of which $\lambda_i$ or $\mu_i$ appeared first in $\eta$, Algorithm~\ref{alg:decomp2} would have created a section with all the $\lambda_i$ and $\mu_i$ in the same iteration. Thus, there cannot be two or more sections which can be merged and therefore there is no $\Delta' < \Delta$.
\end{proof}

\section{Experiments}

In this section, we present the results of our experiments, which have the following goals:

\begin{itemize}
\item Observe the performance of the Plank strategy in various types of networks.
\item Observe how the Plank strategy performs as the number of searchers available increases.
\item Observe how the Plank strategy performs as the size of the network grows.
\item Observe how the Plank strategy performs as we vary the size of the $k$-hubset.
\item Study how the Plank strategy performs as we vary size, number of searchers, and network structure on random DAGs.
\end{itemize}

For the task of computing an FVS, we employ a heuristic introduced in \cite{Eades1993} for computing a feedback arc set. We take the resulting edge set and place a permanent guard on the source node of each edge. However, if the end node of an edge already has a permanent guard we do not need to place an additional permanent guard on its source node.

Finally, we note that for the majority of our datasets the direction of the edges represents a following/trust relation which we reverse to move to an influence relation.

\subsection{Online Networks}

For each of our networks we run the Plank algorithm on the obtained DAG with $s$ ranging from $0.5 - 3\%$ of the size of the network increasing in $0.25\%$ increments. Additionally, we test three $k$-hubset sizes removing 1\%, 3\%, and 5\% of the number of nodes in the network. Then, we plot the number of steps in the resulting search strategy and the ratio of the length of the strategy to the lower bound. In each plot, the blue line represents no $k$-hubset was removed while purple, yellow, and green lines represent $k$-hubsets of size 1\%, 3\%, and 5\% respectively.\\

\begin{figure}[h]
	\centering
	\begin{minipage}{0.45\linewidth}
		\centering
		\includegraphics[width=\textwidth]{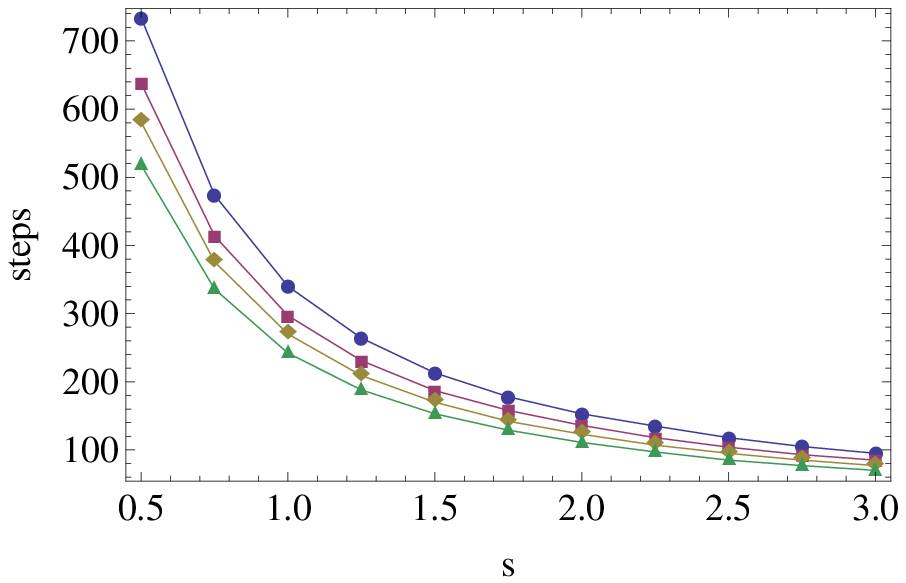}
		\caption{Wiki-Vote strategy lengths.}
		\label{fig:wiki-vote_steps}
	\end{minipage}
	\hspace{0.2cm}
	\begin{minipage}{0.45\linewidth}
		\centering
		\includegraphics[width=\textwidth]{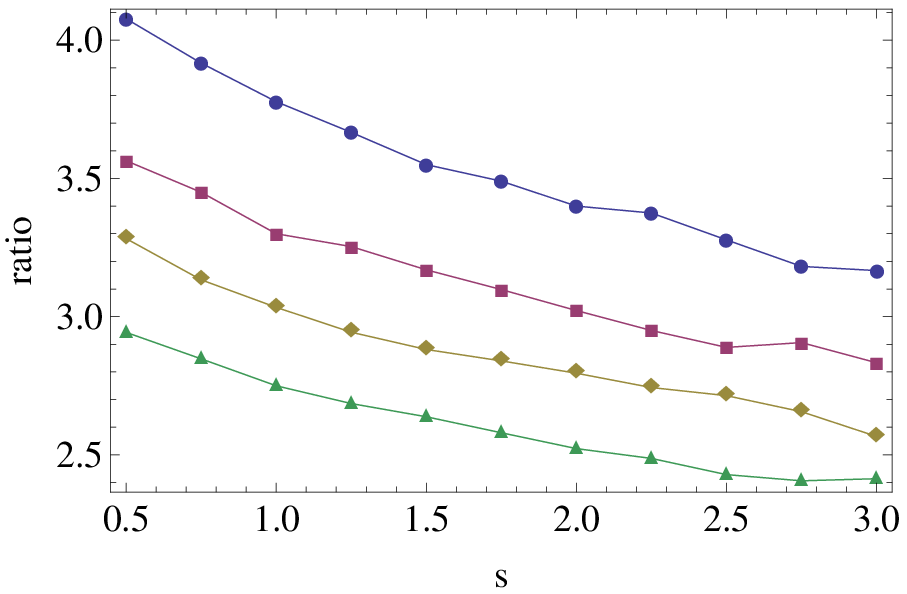}
		\caption{Wiki-Vote approximation ratios.}
		\label{fig:wiki-vote_ratio}
	\end{minipage}
\end{figure}

\begin{figure}[h]
	\centering
	\begin{minipage}{0.45\linewidth}
		\centering
		\includegraphics[width=\textwidth]{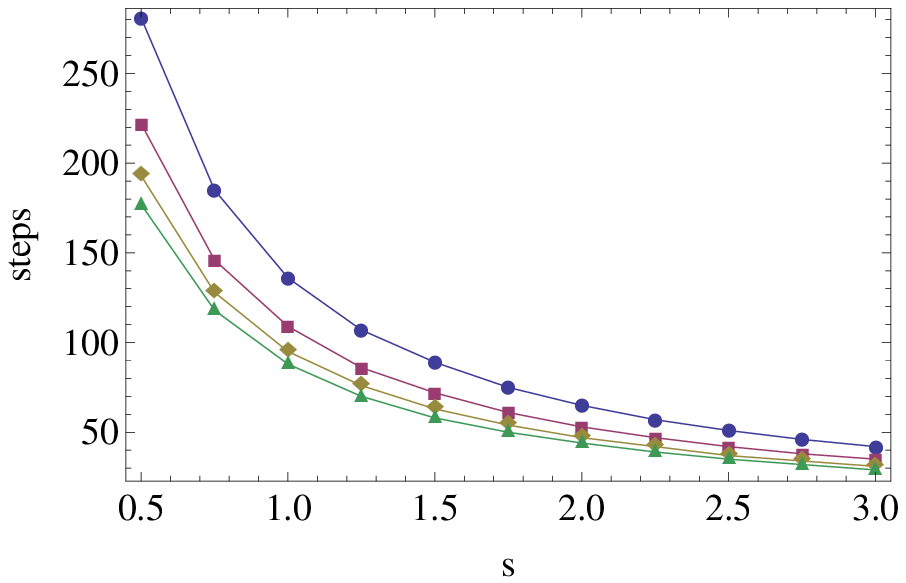}
		\caption{sign-slashdot strategy lengths.}
		\label{fig:sign-slashdot_steps}
	\end{minipage}
	\hspace{0.2cm}
	\begin{minipage}{0.45\linewidth}
		\centering
		\includegraphics[width=\textwidth]{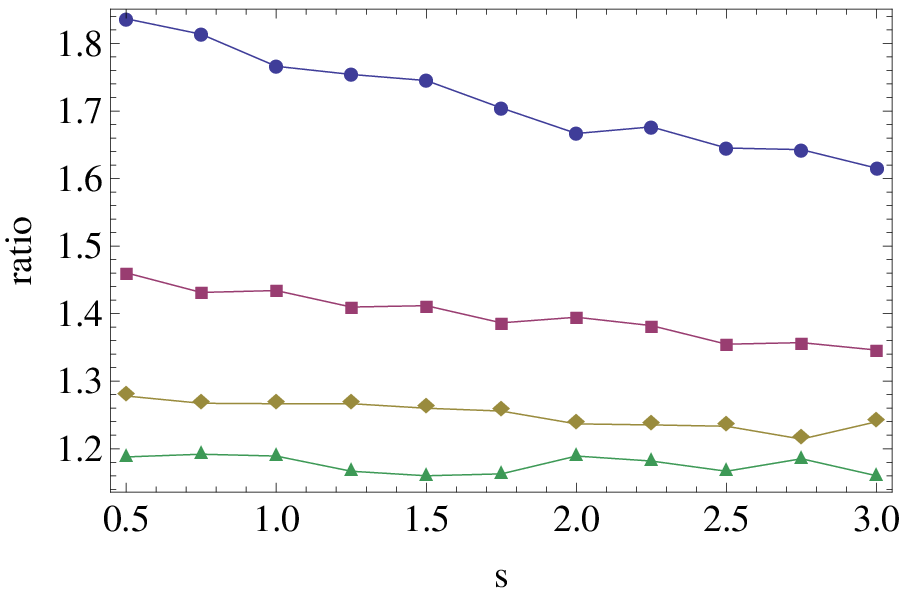}
		\caption{sign-slashdot approximation ratios.}
		\label{fig:sign-slashdot_ratio}
	\end{minipage}
\end{figure}

\begin{figure}[h]
	\centering
	\begin{minipage}{0.45\linewidth}
		\centering
		\includegraphics[width=\textwidth]{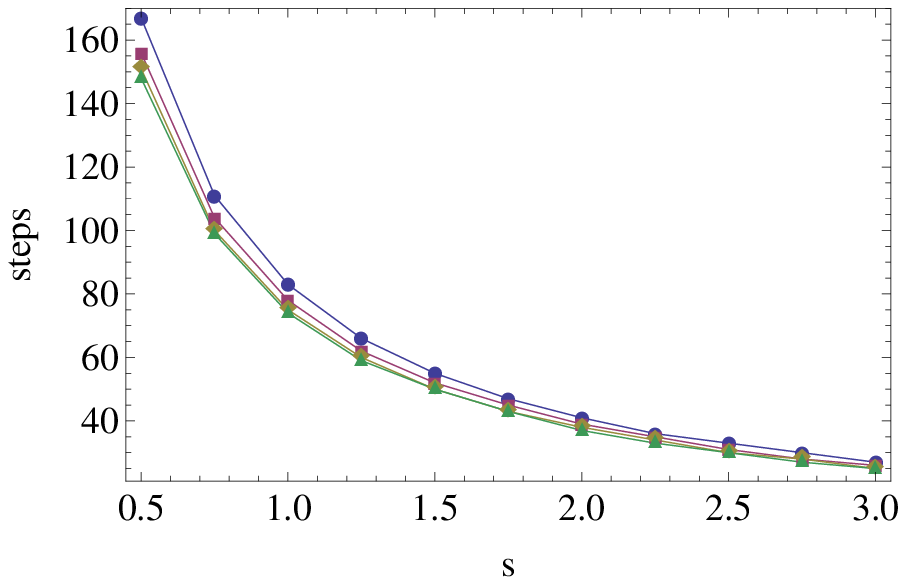}
		\caption{sign-epinions strategy lengths.}
		\label{fig:sign-epinions_steps}
	\end{minipage}
	\hspace{0.2cm}
	\begin{minipage}{0.45\linewidth}
		\centering
		\includegraphics[width=\textwidth]{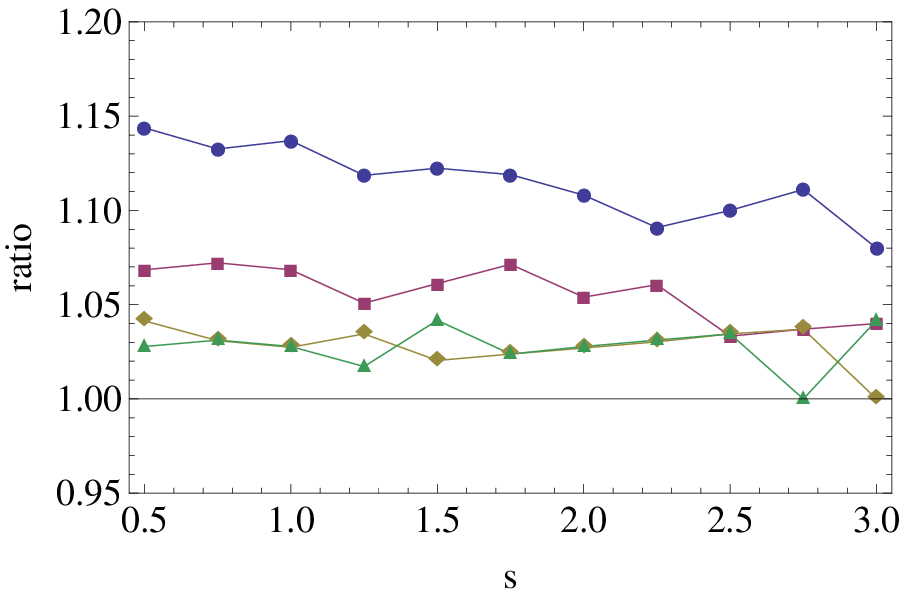}
		\caption{sign-epinions approximation ratios.}
		\label{fig:sign-epinions_ratio}
	\end{minipage}
\end{figure}

\begin{figure}[h]
	\centering
	\begin{minipage}{0.45\linewidth}
		\centering
		\includegraphics[width=\textwidth]{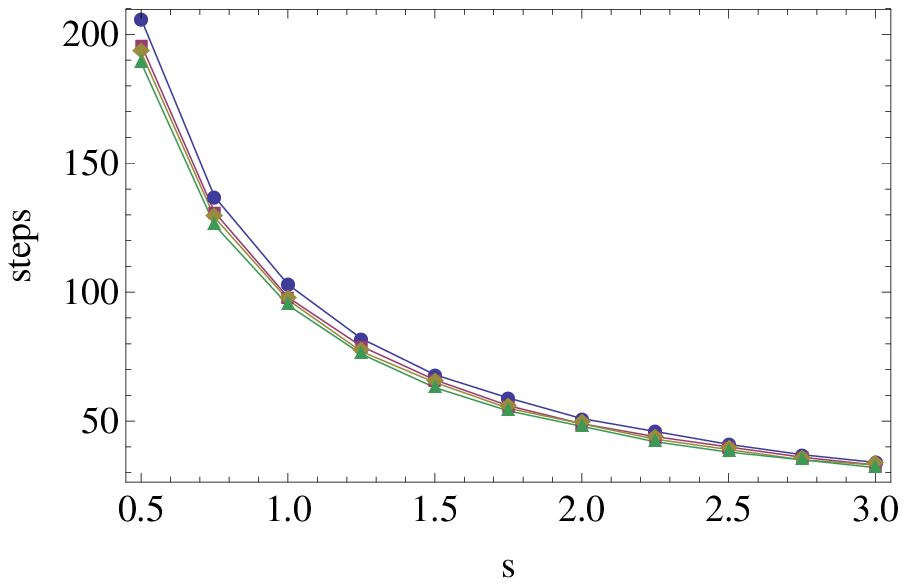}
		\caption{email-EU strategy lengths.}
		\label{fig:email_steps}
	\end{minipage}
	\hspace{0.2cm}
	\begin{minipage}{0.45\linewidth}
		\centering
		\includegraphics[width=\textwidth]{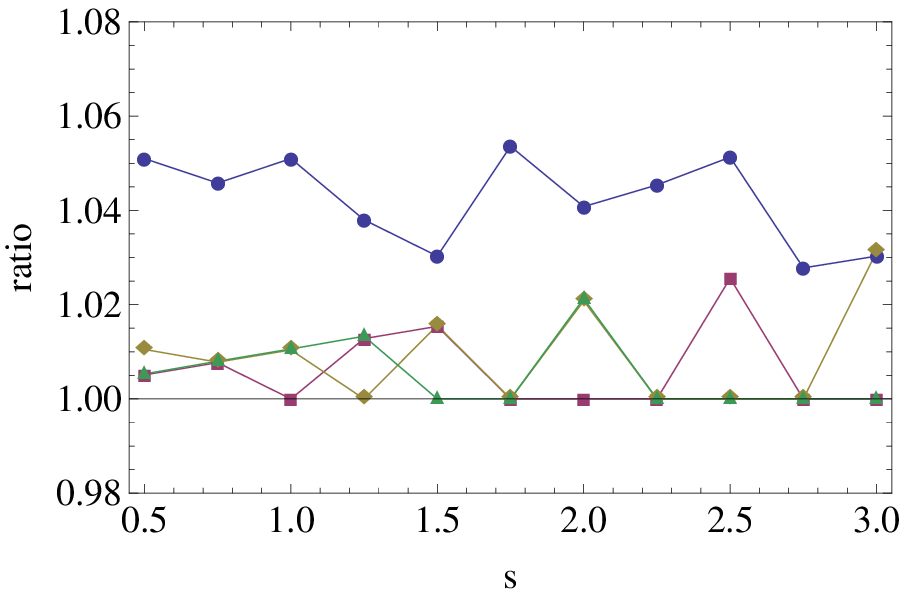}
		\caption{email-EU approximation ratios.}
		\label{fig:email_ratio}
	\end{minipage}
\end{figure}

\begin{figure}[h]
	\centering
	\begin{minipage}{0.45\linewidth}
		\centering
		\includegraphics[width=\textwidth]{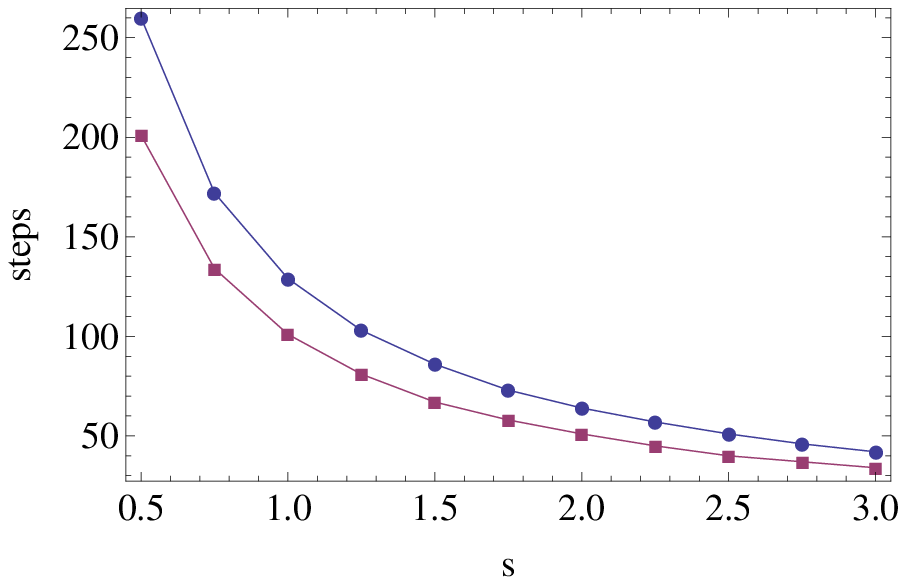}
		\caption{higgs-retweet strategy lengths.}
		\label{fig:higgs-rt_steps}
	\end{minipage}
	\hspace{0.2cm}
	\begin{minipage}{0.45\linewidth}
		\centering
		\includegraphics[width=\textwidth]{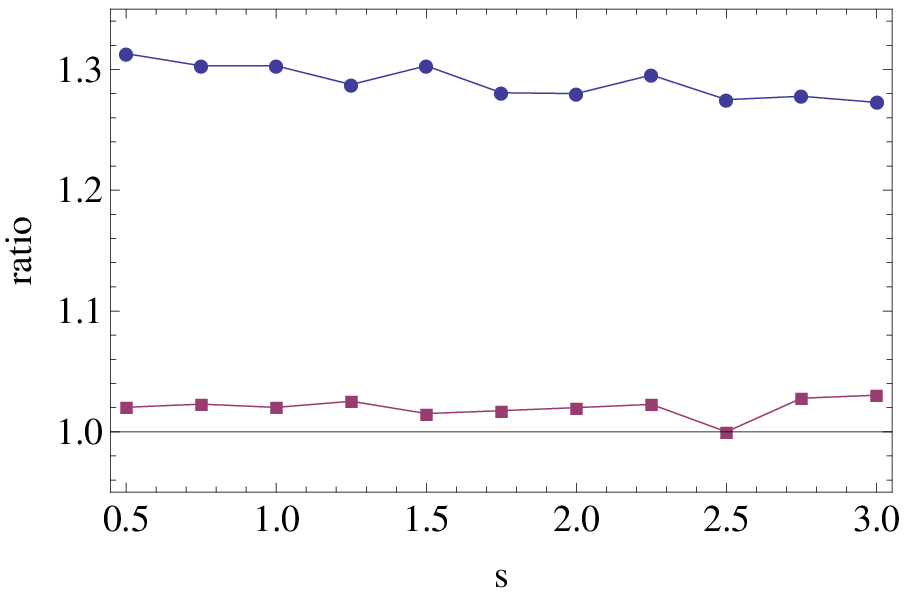}
		\caption{higgs-retweet approximation ratios.}
		\label{fig:higgs-rt_ratio}
	\end{minipage}
\end{figure}

\emph{Wiki-Vote: }First, we look at the Wiki-Vote dataset from \cite{Leskovec2010}. An edge in this network from user A to B indicates that A voted for B to become an administrator. The wiki-vote dataset contains 7,116 nodes and 103,689 edges. The FVS computed contained $11.78 \%$ of the network's nodes.

Figure \ref{fig:wiki-vote_steps} shows the number of steps in the resulting search strategy with increasing number of searchers. Figure \ref{fig:wiki-vote_ratio} plots the approximation ratio versus the number of searchers. The approximation ratio drops steadily in each case with larger $k$-hubsets performing better. %The length drops from $734$ to $95$ steps while the approximation ratio drops from $4.08$ to $3.17$. \\

\emph{Signed Slashdot: }Next, we look at the signed Slashdot dataset from \cite{Leskovec2010}. An edge in this network from user A to B indicates that B is a friend of A's. The signed Slashdot dataset contains 77.350 nodes and 516,575 edges. The FVS computed contained $16.46 \%$ of the network's nodes.

Figure \ref{fig:sign-slashdot_steps} shows the number of steps in the resulting search strategy with increasing number of searchers. Figure \ref{fig:sign-slashdot_ratio} plots the approximation ratio versus the number of searchers. Again, we see that larger $k$-hubsets produce better approximation ratios. %The length drops from $281$ to $42$ steps while the approximation ratio drops from $1.84$ to $1.61$. \\

\emph{Signed Epinions: }The signed Epinions trust network from \cite{Leskovec2010} contains an edge from user A to B if A trusts B on the Epinions review site. The signed Epinions dataset contains 131,828 nodes and 841,372 edges. The FVS computed contained $16.04 \%$ of the network's nodes.

Figures \ref{fig:sign-epinions_steps} and \ref{fig:sign-epinions_ratio} show the number of steps in the resulting search strategy and approximation ratio respectively with increasing number of searchers. Here we see that as a near optimal approximation ratio is approached the larger $k$-hubsets lose their effect. %The length drops from $167$ to $27$ steps while the approximation ratio drops from $1.14$ to $1.08$. \\

\emph{Email Communication Network: }The email-EU network from \cite{Leskovec2007} was generated using email data from a large European research institution. The network contains an edge from user A to B if A emailed B. The network contains 265,214 nodes and 420,045 edges. The FVS computed contained only $2.45 \%$ of the network's nodes indicating a very DAG-like structure.

The number of steps required by the search strategy and approximation ratio are shown in figures \ref{fig:email_steps} and \ref{fig:email_ratio} respectively. For this dataset the $k$-hubset cases all hover near an optimal approximation ratio while the case with no $k$-hubset sits slightly above. %The length drops from $206$ to $34$ steps while the approximation ratio remains close to optimal. \\

\emph{Twitter Retweet Network: }The higgs-retweet network from \cite{Domenico2013} maps the retweets by users of Twitter during the announcement of the discovery of the Higgs Boson. The network contains an edge from user A to B if A retweeted B. The network contains 425,008 nodes and 733,647 edges. The FVS computed contained $1.13 \%$ of the network's nodes also indicating a very DAG-like structure.

Figure \ref{fig:higgs-rt_steps} shows the number of steps in the resulting search strategy with increasing number of searchers and Figure \ref{fig:higgs-rt_ratio} plots the approximation ratio versus the number of searchers. Here we only include the 1\% $k$-hubset case since, as we saw with the email-EU dataset, larger $k$-hubsets don't provide any improvement once we near an optimal approximation ratio. \\ %The length drops from $260$ to $42$ steps while the approximation ratio drops from $1.31$ to $1.27$.

Next, we plot how well the sliding FVS saves searchers compared to the size of the FVS computed for each dataset. The Wiki-Vote dataset did not benefit from the sliding FVS most likely due to its high density compared to the other networks. Then, as we see in Figure \ref{fig:wiki-vote_fvs}, the removal of each $k$-hubset only increases the size of the FVS with increasing $k$-hubset size. Figure \ref{fig:sign-slashdot_fvs} shows searcher savings in all cases for the signed Slashdot network, and interestingly, the 1\% $k$-hubset produces the minimum sliding FVS. Then, Figure \ref{fig:sign-epinions_fvs} shows that the searcher savings increase with increasing $k$-hubset size in the signed Epinions network. Similar to the signed Slashdot network, Figure \ref{fig:email_fvs} shows an optimal sliding FVS when using the 1\% $k$-hubset for the email-EU network. Figure \ref{fig:higgs-rt_fvs} shows a sliding FVS increasing in size with a larger $k$-hubset for the higgs-retweet network. The sliding FVS results for the higgs-retweet network can be fairly easily predicted from the fact that the FVS computed was already very small and thus any permanent guards enforced by the $k$-hubset exceed this value.

\begin{figure}[h]
	\centering
	\begin{minipage}{0.45\linewidth}
		\centering
		\includegraphics[width=\textwidth]{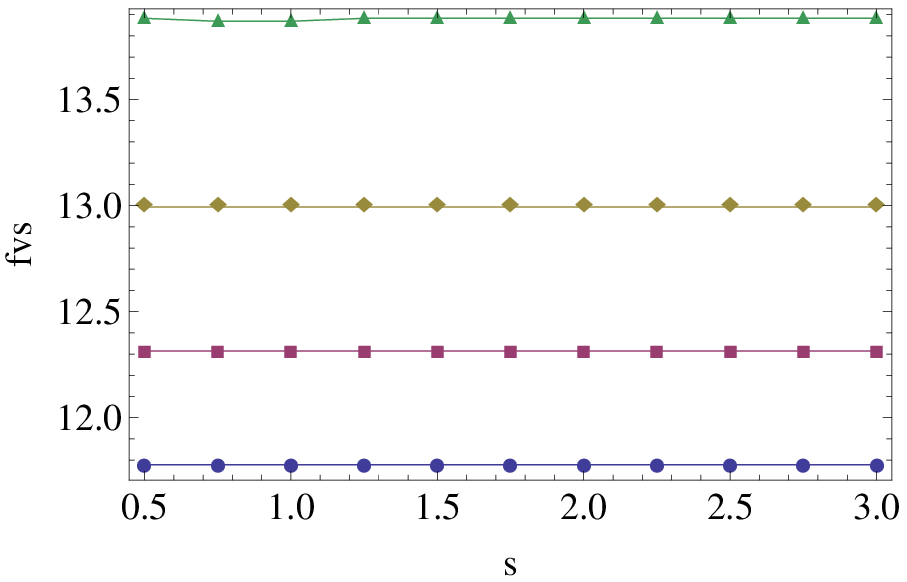}
		\caption{Wiki-Vote sliding FVS.}
		\label{fig:wiki-vote_fvs}
	\end{minipage}
	\hspace{0.2cm}
	\begin{minipage}{0.45\linewidth}
		\centering
		\includegraphics[width=\textwidth]{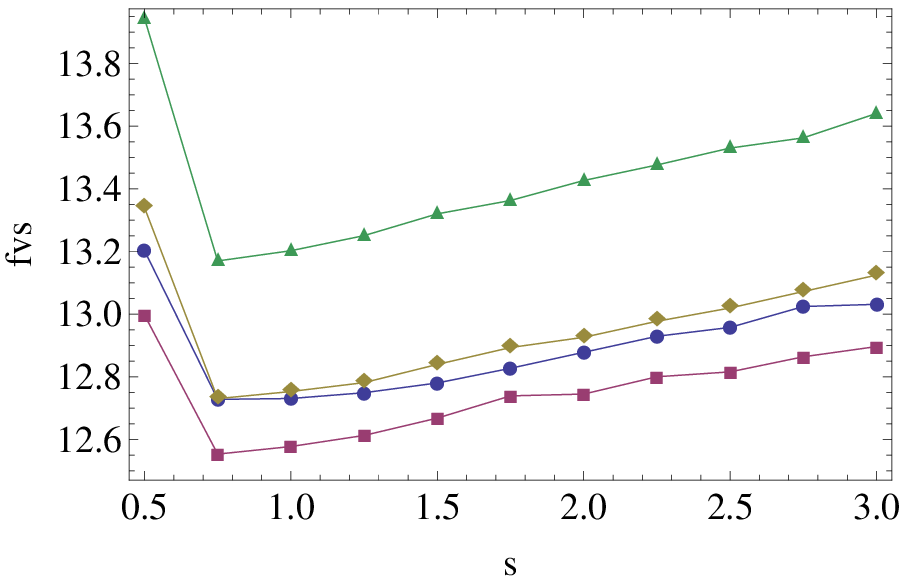}
		\caption{sign-slashdot sliding FVS.}
		\label{fig:sign-slashdot_fvs}
	\end{minipage}
\end{figure}

\begin{figure*}
	\centering
	\begin{minipage}{0.31\linewidth}
		\centering
		\includegraphics[width=\textwidth]{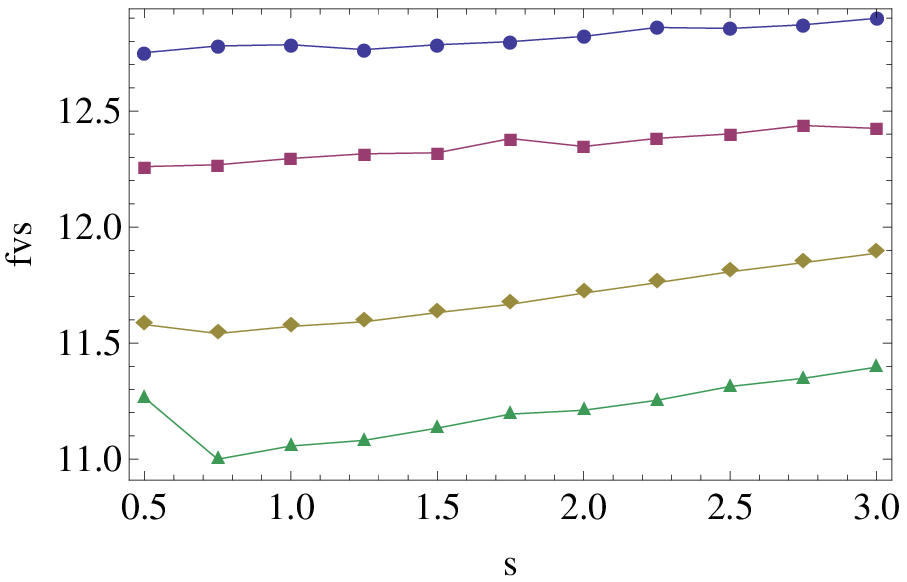}
		\caption{sign-epinions sliding FVS.}
		\label{fig:sign-epinions_fvs}
	\end{minipage}
	\hspace{0.2cm}
	\begin{minipage}{0.31\linewidth}
		\centering
		\includegraphics[width=\textwidth]{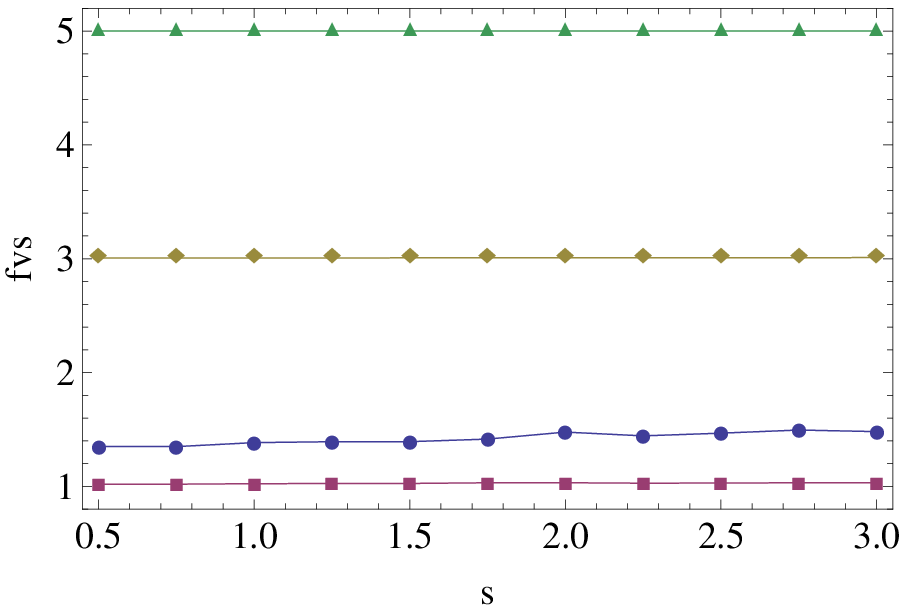}
		\caption{email-EU sliding FVS.}
		\label{fig:email_fvs}
	\end{minipage}
	\hspace{0.2cm}
	\begin{minipage}{0.31\linewidth}
		\centering
		\includegraphics[width=\textwidth]{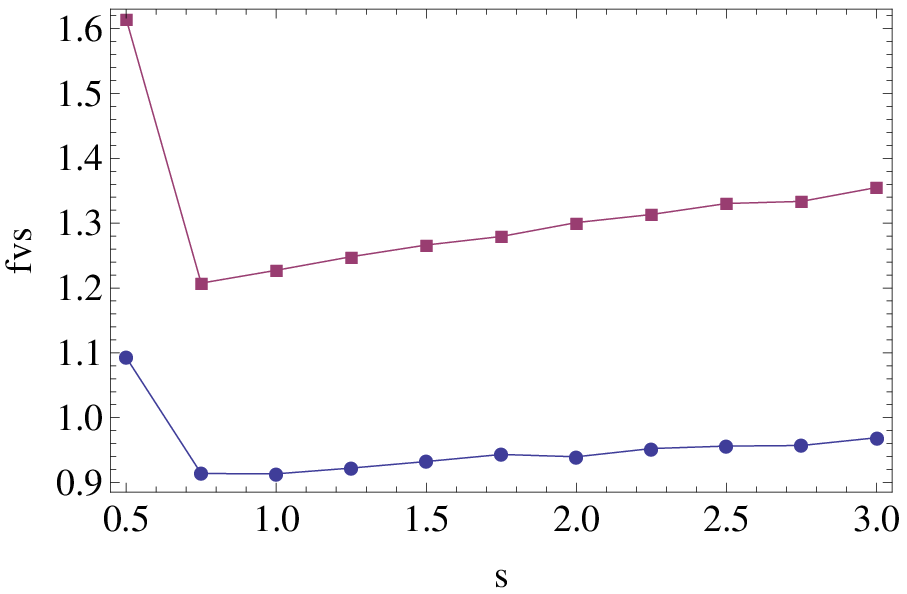}
		\caption{higgs-retweet sliding FVS.}
		\label{fig:higgs-rt_fvs}
	\end{minipage}
\end{figure*}

We note that the regularity in each of the plots showing the length of a strategy indicates that the potential loss between steps when there is a leftover searcher unable to clear an additional edge does not have a large effect on the number of steps required to clear the network.  Furthermore, near optimal approximation ratios indicate the DAG remaining after the removal of the FVS had a small overlap value. Finally, the Wiki-Vote dataset is the only network in which the overlap factor drove the approximation ratio above $2$. However, we see a good decrease ranging from $17.99$ - $22.34\%$ in the approximation ratio.

\subsection{Random DAGs}

Next, we consider the individual parameters of the system and investigate how the approximation ratio is affected as they are varied. For these tests, we generate random DAGs similar to the Erd\"{o}s-R\'{e}nyi model except we predetermine an ordering of the nodes, $(1 \ldots n)$, in the DAG and then randomly add edges from node $i$ to $j$ with probability $p$ provided $i$ comes before $j$ in the ordering. We generate five random DAGs for each data point and average the results.

First, we study how the approximation ratio behaves as the size of the network is increased. We fix $p=\frac{1}{n}$ and run the tests for $s=10, 25, 50$. Figure \ref{fig:random_size} shows the resulting approximation ratios as the network size increases from 1,000 to 20,000 nodes. We observe the ratios remain nearly constant as the network size is increased.

Next, we look at how the approximation ratio behaves as the number of searchers is increased. We fix $p=\frac{1}{n}$ and run the tests for $n=$ 5,000, $n=$ 10,000, and $n=$ 20,000. Figure \ref{fig:random_num} shows the resulting approximation ratios as the number of searchers increases from 0.2\% to 2\% of the network size. We see that the approximation ratio decreases as the number of searchers increases. \\

\begin{figure}[h]
	\centering
	\begin{minipage}{0.45\linewidth}
		\centering
		\includegraphics[width=\textwidth]{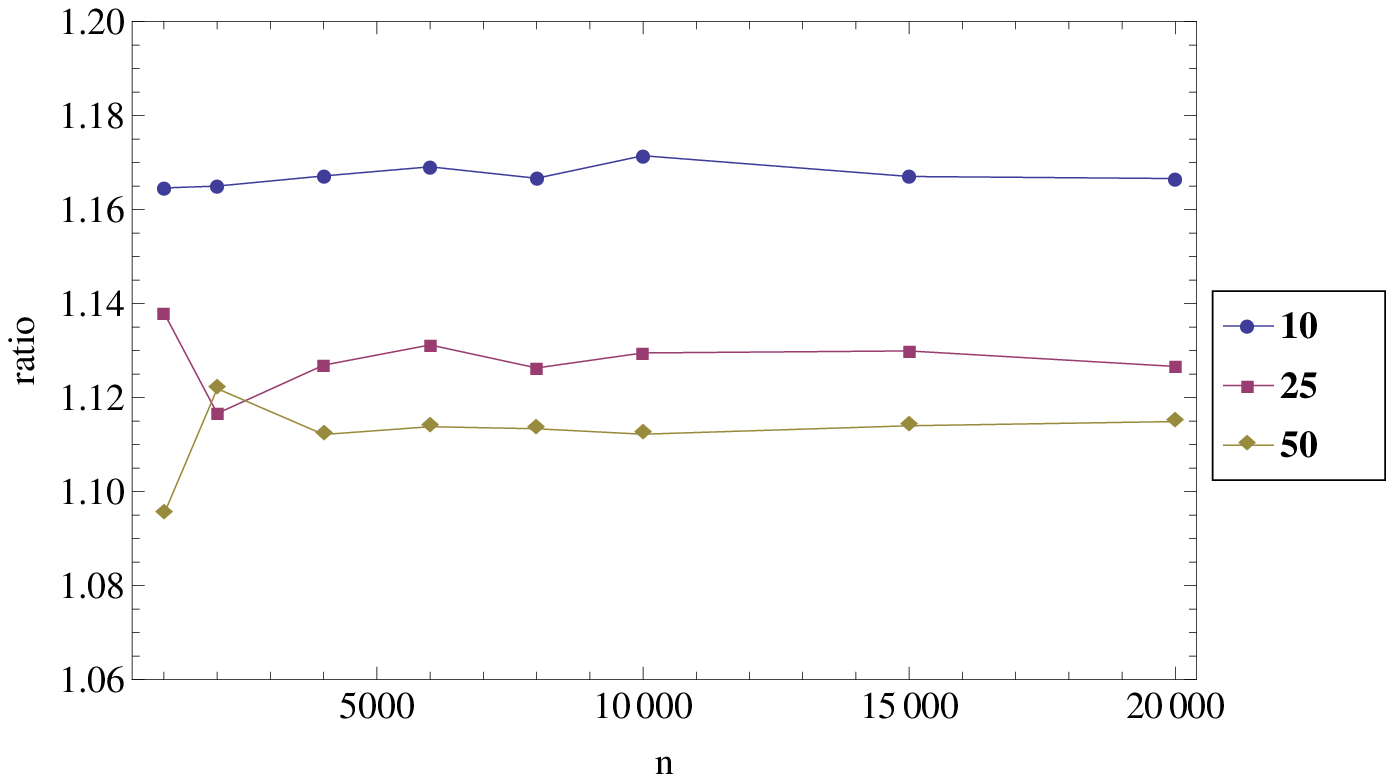}
		\caption{Effect of increasing network size on the approximation ratio.}
		\label{fig:random_size}
	\end{minipage}
	\hspace{0.2cm}
	\begin{minipage}{0.45\linewidth}
		\centering
		\includegraphics[width=\textwidth]{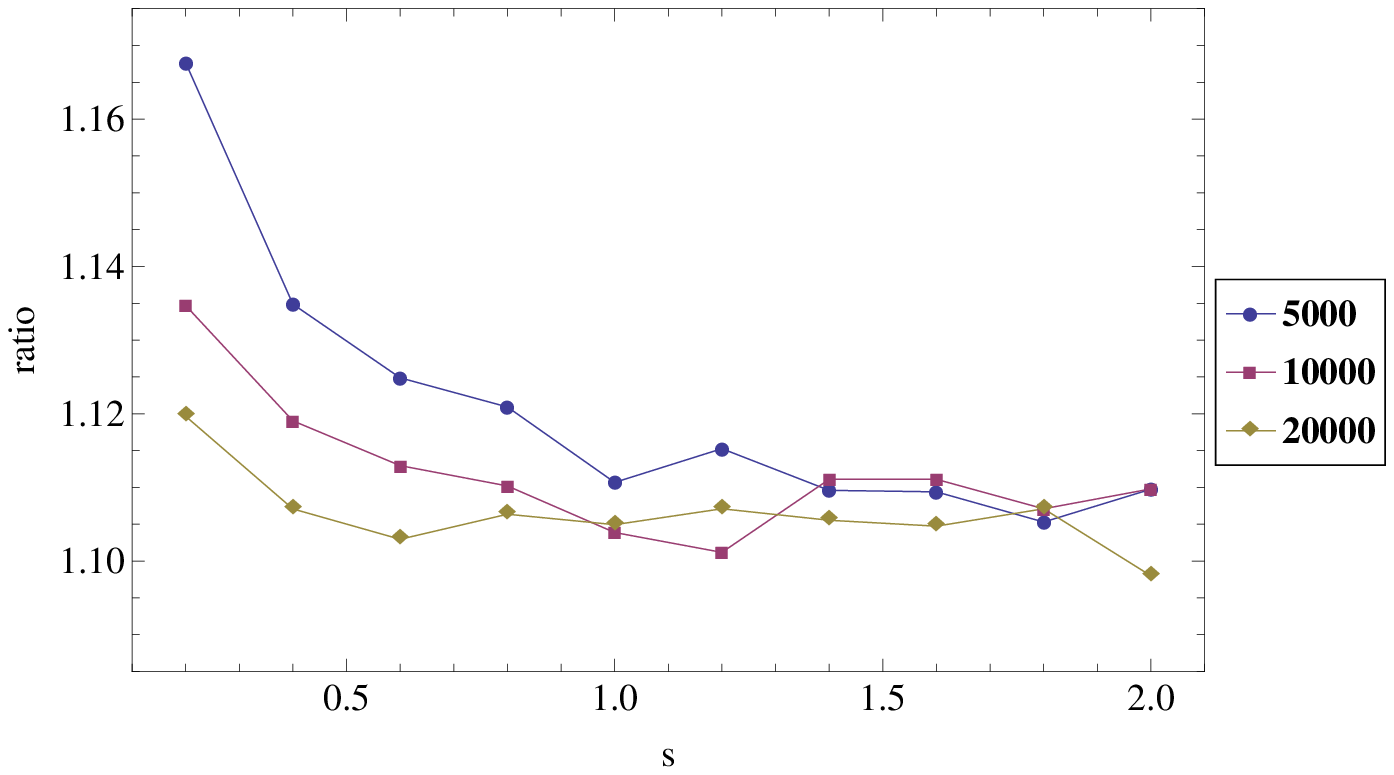}
		\caption{Effect of increasing searchers on the approximation ratio.}
		\label{fig:random_num}
	\end{minipage}
\end{figure}

\makebox[0.96\linewidth][s]{Futhermore, we produce random DAGs according to the} \\*Barab\'{a}si-Albert model \cite{Barabasi99} to replicate the power law structure exhibited in many online social networks. The Barab\'{a}si-Albert model takes three parameters $n$, $m$, and $m_0$. The graph begins with $m_0$ isolated nodes. New nodes are added to the graph one at a time until we have a graph with $n$ nodes. Each new node is connected to $m \le m_0$ existing nodes with a probability that is proportional to the number of edges that the existing nodes already have. We direct new edges from existing nodes to new nodes to maintain a DAG structure. We run the Plank algorithm on each DAG with $n=$ 20,000 and $s$ ranging from $0.5 - 3\%$ of the size of the network increasing in $0.25\%$ increments.

\begin{figure*}
	\centering
	\begin{minipage}{0.31\linewidth}
		\centering
		\includegraphics[width=\textwidth]{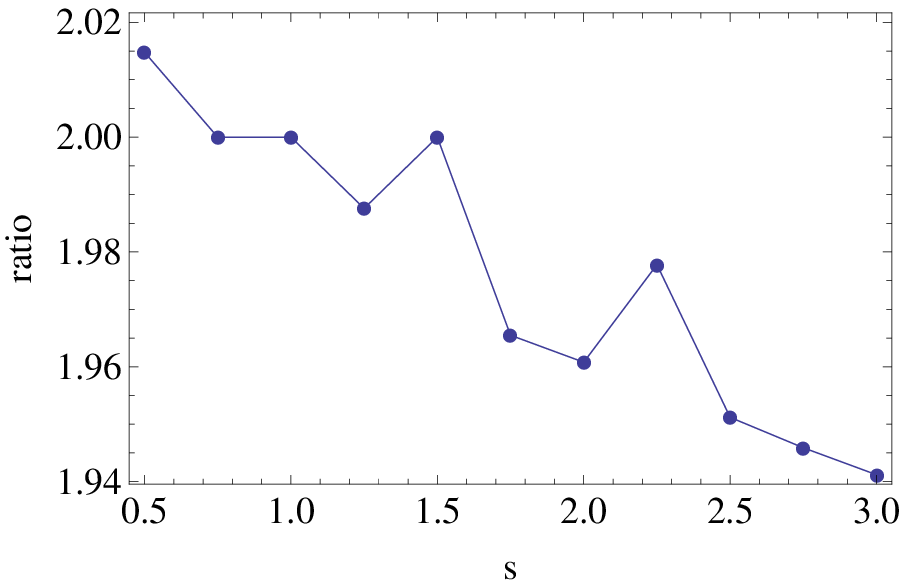}
		\caption{$m = m_0 = 2$.}
		\label{fig:random_pref_2_2}
	\end{minipage}
	\hspace{0.2cm}
	\begin{minipage}{0.31\linewidth}
		\centering
		\includegraphics[width=\textwidth]{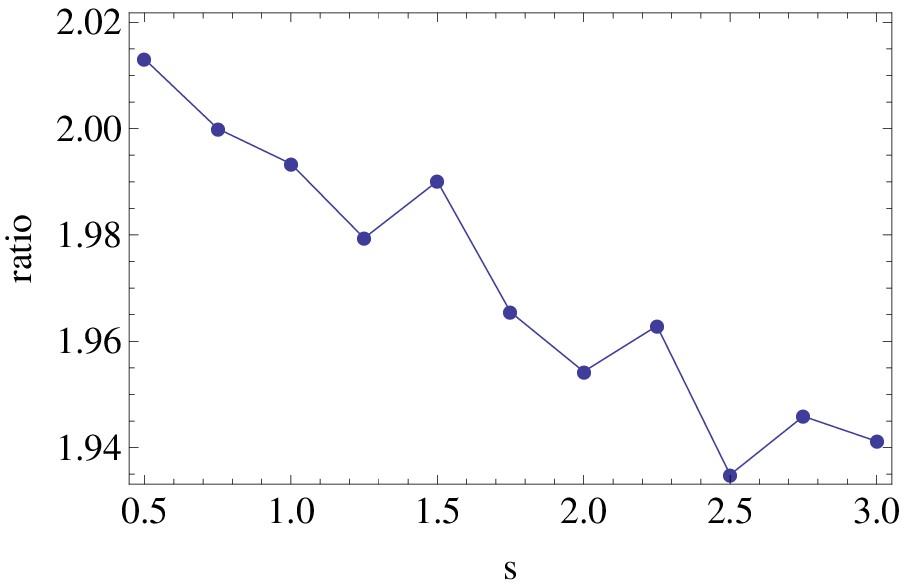}
		\caption{$m = 2$, $m_0 = 3$.}
		\label{fig:random_pref_2_3}
	\end{minipage}
	\hspace{0.2cm}
	\begin{minipage}{0.31\linewidth}
		\centering
		\includegraphics[width=\textwidth]{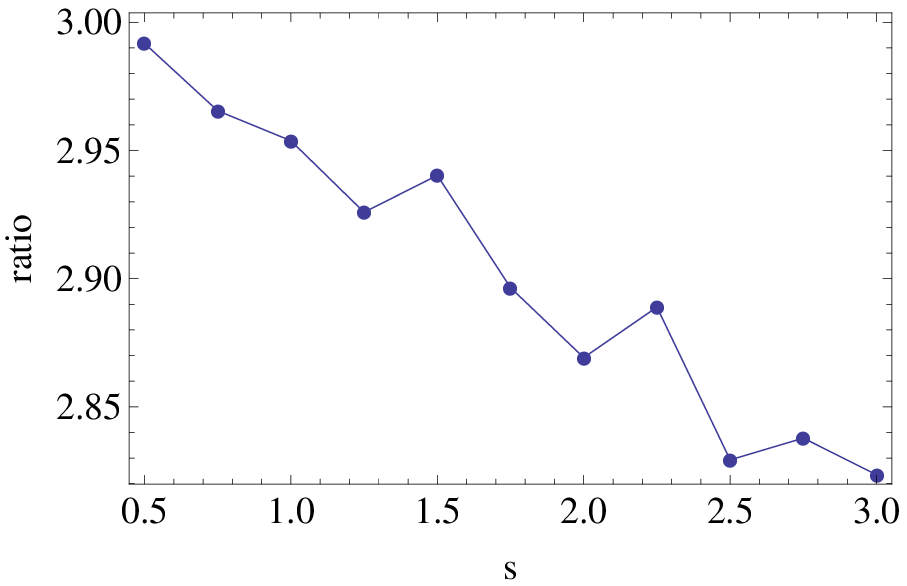}
		\caption{$m = 3$, $m_0 = 6$.}
		\label{fig:random_pref_3_6}
	\end{minipage}
	\caption{Barab\'{a}si-Albert DAG plots.}
\end{figure*}

In Figure \ref{fig:random_pref_2_2} we have $m = m_0 = 2$ and see a steady decrease in approximation ratio. Then, in Figure \ref{fig:random_pref_2_3} we investigate the effects of adding an additional preferential node where we observe a similar decrease in approximation ratio. Finally, in Figure \ref{fig:random_pref_3_6} we look at a Barab\'{a}si-Albert DAG in which there are $6$ preferential nodes and $3$ links are added with each new node in which a decreasing approximation ratio is also observed as the number of searchers (modestly) increases. In each case we observe good approximation ratios.

%Finally, we study how the Plank strategy performs as the density of the DAG is increased. We fix $n = 5,000$ and $s = 50$ and vary $p$ from $\frac{1}{2n}$ to $\frac{14}{n}$. Figure \ref{fig:random_density} shows the resulting approximation ratios which increases linearly as the expected number of edges incident on each node increases. This captures the potential loss we discussed induced by the overlap factor.

%\begin{figure}
%	\centering
%	\includegraphics[scale=0.5]{Images/random_density.eps}
%	\caption{Effect of increasing network density.}
%	\label{fig:random_density}
%\end{figure}

\section{Conclusion}

In this work we perform an extensive study of the problem of eliminating contamination spreading through a network. Specifically, we study the related graph searching problem which we prove is NP-hard even on DAGs and therefore an exact algorithm is infeasible for large networks. Consequently, we introduce a novel approximation algorithm for clearing DAGs which we incorporate into a procedure for clearing general digraphs. We experimentally test our algorithm on several large online networks and observe good performance in relation to the lower bound. Furthermore, we explore various parameters of the graph searching problem on random DAGs and discover the search time is unaffected by network size, yet significantly decreases with modest increases in searcher allocation.

\balance

\bibliographystyle{abbrv}
\bibliography{/Users/michaelesimp/Documents/UVic/Research/refs}

\begin{thebibliography}{10}

\bibitem{Barabasi99}
A.-L. Barabasi and R.~Albert.
\newblock Emergence of scaling in random networks.
\newblock {\em Science}, 286(5439):509--512, 1999.

\bibitem{Bharathi2007}
S.~Bharathi, D.~Kempe, and M.~Salek.
\newblock Competitive influence maximization in social networks.
\newblock In {\em Proceedings of the 3rd international conference on Internet
  and network economics}, WINE'07, pages 306--311, Berlin, Heidelberg, 2007.
  Springer-Verlag.

\bibitem{Bienstock1991}
D.~Bienstock.
\newblock Graph searching, path-width, tree-width and related problems.
\newblock {\em DIMACS Ser. in Discrete Mathematics and Theoretical Computer
  Science}, 5:33--49, 1991.

\bibitem{Blin2008DCN}
L.~Blin, P.~Fraigniaud, N.~Nisse, and S.~Vial.
\newblock Distributed chasing of network intruders.
\newblock {\em Theor. Comput. Sci.}, 399(1-2):12--37, June 2008.

\bibitem{Borie2011}
R.~Borie, C.~Tovey, and S.~Koenig.
\newblock Algorithms and complexity results for graph-based pursuit evasion.
\newblock {\em Auton. Robots}, 31(4):317--332, Nov. 2011.

\bibitem{Brandenburg2006UQ}
F.~J. Brandenburg and S.~Herrmann.
\newblock Graph searching and search time.
\newblock In {\em SOFSEM 2006: Theory and Practice of Computer Science}, volume
  3831 of {\em Lecture Notes in Computer Science}, pages 197--206. 2006.

\bibitem{budak2011limiting}
C.~Budak, D.~Agrawal, and A.~El~Abbadi.
\newblock Limiting the spread of misinformation in social networks.
\newblock In {\em Proceedings of the 20th international conference on World
  wide web}, pages 665--674. ACM, 2011.

\bibitem{Carnes2007}
T.~Carnes, C.~Nagarajan, S.~M. Wild, and A.~van Zuylen.
\newblock Maximizing influence in a competitive social network: a follower's
  perspective.
\newblock In {\em Proceedings of the ninth international conference on
  Electronic commerce}, ICEC '07, pages 351--360, New York, NY, USA, 2007. ACM.

\bibitem{Chen2013}
W.~Chen, L.~V.~S. Lakshmanan, and C.~Castillo.
\newblock {\em Information and Influence Propagation in Social Networks}.
\newblock Synthesis Lectures on Data Management. Morgan {\&} Claypool
  Publishers, 2013.

\bibitem{Chen2010}
W.~Chen, Y.~Yuan, and L.~Zhang.
\newblock Scalable influence maximization in social networks under the linear
  threshold model.
\newblock In {\em Proceedings of the 2010 IEEE International Conference on Data
  Mining}, ICDM '10, pages 88--97, Washington, DC, USA, 2010. IEEE Computer
  Society.

\bibitem{Dendris1994}
N.~D. Dendris, L.~M. Kirousis, and D.~M. Thilikos.
\newblock Fugitive-search games on graphs and related parameters.
\newblock In {\em Proceedings of the 20th International Workshop on
  Graph-Theoretic Concepts in Computer Science}, WG '94, pages 331--342,
  London, UK, UK, 1995. Springer-Verlag.

\bibitem{Domenico2013}
M.~D. Domenico, A.~Lima, P.~Mougel, and M.~Musolesi.
\newblock The anatomy of a scientific rumor.
\newblock {\em Scientific Reports}, 3, 01 2013.

\bibitem{Eades1993}
P.~Eades, X.~Lin, and W.~F. Smyth.
\newblock A fast and effective heuristic for the feedback arc set problem.
\newblock {\em Inf. Process. Lett.}, 47(6):319--323, Oct. 1993.

\bibitem{Ellis1994}
J.~A. Ellis, I.~H. Sudborough, and J.~S. Turner.
\newblock The vertex separation and search number of a graph.
\newblock {\em Inf. Comput.}, 113(1):50--79, Aug. 1994.

\bibitem{Goyal2013}
A.~Goyal, F.~Bonchi, L.~V.~S. Lakshmanan, and S.~Venkatasubramanian.
\newblock On minimizing budget and time in influence propagation over social
  networks.
\newblock {\em Social Netw. Analys. Mining}, 3(2):179--192, 2013.

\bibitem{Kang2011}
U.~Kang and C.~Faloutsos.
\newblock Beyond 'caveman communities': Hubs and spokes for graph compression
  and mining.
\newblock In {\em Proceedings of the 2011 IEEE 11th International Conference on
  Data Mining}, ICDM '11, pages 300--309, Washington, DC, USA, 2011. IEEE
  Computer Society.

\bibitem{Kirousis1986SP}
M.~Kirousis and C.~H. Papadimitriou.
\newblock Searching and pebbling.
\newblock {\em Theor. Comput. Sci.}, 47(2):205--218, Nov. 1986.

\bibitem{Leskovec2010}
J.~Leskovec, D.~Huttenlocher, and J.~Kleinberg.
\newblock Predicting positive and negative links in online social networks.
\newblock In {\em Proceedings of the 19th International Conference on World
  Wide Web}, WWW '10, pages 641--650, New York, NY, USA, 2010. ACM.

\bibitem{Leskovec2007}
J.~Leskovec, J.~Kleinberg, and C.~Faloutsos.
\newblock Graph evolution: densification and shrinking diameters.
\newblock {\em ACM Trans. Knowl. Discov. Data}, 1(1), Mar. 2007.

\bibitem{Liu2012}
B.~Liu, G.~Cong, D.~Xu, and Y.~Zeng.
\newblock Time constrained influence maximization in social networks.
\newblock In {\em ICDM}, pages 439--448, 2012.

\bibitem{Megiddo1988}
N.~Megiddo, S.~L. Hakimi, M.~R. Garey, D.~S. Johnson, and C.~H. Papadimitriou.
\newblock The complexity of searching a graph.
\newblock {\em J. ACM}, 35(1):18--44, Jan. 1988.

\bibitem{Meier2008}
D.~Meier, Y.~A. Oswald, S.~Schmid, and R.~Wattenhofer.
\newblock On the windfall of friendship: inoculation strategies on social
  networks.
\newblock In {\em Proceedings of the 9th ACM Conference on Electronic
  Commerce}, EC '08, pages 294--301, New York, NY, USA, 2008. ACM.

\bibitem{Parsons1976}
Y.~A. Parsons, T.D. and D.~Lick.
\newblock Pursuit-evasion in a graph.
\newblock {\em Theory and Applications of Graphs}, pages 426--441, 1976.

\end{thebibliography}

\end{document}